\def\paperversion{arxiv}     %

\def\ppsarxiv{arxiv}
\newcommand{\preprintswitch}[2]{%
  \ifx\paperversion\ppsarxiv
    \edef\ppstmp{\detokenize{#2}}\ifx\ppstmp\empty\ifhmode\unskip\fi\else#2\fi
  \else
    \edef\ppstmp{\detokenize{#1}}\ifx\ppstmp\empty\ifhmode\unskip\fi\else#1\fi
  \fi}

\PassOptionsToPackage{dvipsnames,table}{xcolor}

\preprintswitch{
\documentclass[journal,twoside,web]{ieeecolor}  %
}
{
\documentclass[journal]{IEEEtran} %
}

\usepackage{generic}
\usepackage{amsmath}
\preprintswitch{}{\usepackage{amsthm}}
\usepackage{amssymb}
\usepackage{mathtools}
\usepackage{bm}
\usepackage{nicefrac}

\usepackage{graphicx}
\usepackage{subcaption}

\usepackage{makecell}
\usepackage{arydshln}
\usepackage{nicematrix}

\usepackage{algorithm}
\usepackage[noend]{algpseudocode}

\usepackage{cite}
\usepackage{hyperref}
\usepackage{cleveref}

\usepackage{siunitx}
\usepackage{enumerate}

\usepackage{tikz}

\usepackage{globalvals}

\newtheorem{definition}{Definition}
\newtheorem{assumption}{Assumption}
\newtheorem{theorem}{Theorem}
\newtheorem{proposition}{Proposition}
\newtheorem{lemma}{Lemma}
\newtheorem{remark}{Remark}
\newtheorem{corollary}{Corollary}
\newtheorem{example}{Example}

\preprintswitch{
\newcommand{\qed}{\unskip\nobreak\hfill $\square$}
}
{}
\let\oldexample\example
\let\oldendexample\endexample
\def\example{\begingroup \oldexample}
\def\endexample{\qed \oldendexample \endgroup}

\newcommand\real[1]{\mathbb{R}^{#1}}
\newcommand{\tf}[1]{\mathbf{#1}}
\newcommand{\ttf}[1]{\boldsymbol{#1}}
\newcommand{\spreg}{\operatorname{SpReg}}

\DeclareMathOperator*{\st}{s.t.}
\DeclareMathOperator*{\argmin}{arg\,min}

\DeclareMathOperator*{\H2}{\mathcal{H}_2}
\DeclareMathOperator*{\Hinfty}{\mathcal{H}_\infty}
\newcommand{\norm}[1]{\left\lVert#1\right\rVert}
\newcommand{\abs}[1]{\left|#1\right|}

\DeclareMathOperator*{\graph}{\mathcal{G}}
\DeclareMathOperator*{\vertices}{\mathcal{V}}
\DeclareMathOperator*{\edges}{\mathcal{E}}
\DeclareMathOperator*{\adjmat}{\mathcal{A}(\graph)}
\DeclareMathOperator*{\netstruc}{\mathtt{NS}}

\def\ppsjournalblue{journalblue}
\ifx\paperversion\ppsjournalblue
  \newcommand{\blue}[1]{\edef\bluetmp{\detokenize{#1}}\ifx\bluetmp\empty\ifhmode\unskip\fi\else\textcolor{NavyBlue}{#1}\fi}
\else
  \newcommand{\blue}[1]{\edef\bluetmp{\detokenize{#1}}\ifx\bluetmp\empty\ifhmode\unskip\fi\else#1\fi}
\fi

\definecolor{lightgray}{gray}{0.9}
\definecolor{lightgrayshblue}{RGB}{211,229,244}

\usepackage{dsfont}

\preprintswitch{
\def\BibTeX{{\rm B\kern-.05em{\sc i\kern-.025em b}\kern-.08em
    T\kern-.1667em\lower.7ex\hbox{E}\kern-.125emX}}
\markboth{\journalname, VOL. XX, NO. XX, XXXX 2017}
{Author \MakeLowercase{\textit{et al.}}: Preparation of Brief Papers for IEEE TRANSACTIONS and JOURNALS (February 2017)}
}
{}

\begin{document}

\title{\LARGE \bf   
A graph-informed regret metric for optimal distributed control
}

\author{Daniele Martinelli, Andrea Martin, Giancarlo Ferrari-Trecate, and Luca Furieri%
\thanks{D. Martinelli and G. Ferrari-Trecate are with the Institute of Mechanical Engineering, EPFL, Switzerland (\texttt{\{daniele.martinelli, giancarlo.ferraritrecate\}@epfl.ch}).
Andrea Martin is with the School of Electrical Engineering and Computer Science, and Digital Futures, KTH Royal Institute of Technology, Sweden (\texttt{andrmar@kth.se}). Luca Furieri is with the Department of Engineering Science, University of Oxford, United Kingdom (\texttt{luca.furieri@eng.ox.ac.uk}).}
\thanks{ The authors acknowledge funding from the Swiss National Science Foundation (SNSF) for the Ambizione grant PZ00P2\textunderscore208951, the NCCR Automation (grant agreement 51NF40\textunderscore 225155) and the NECON project (grant 200021\_219431).
}
}

\maketitle
\begin{abstract}
We consider the optimal control of large-scale systems using distributed controllers whose network topology mirrors the coupling graph between subsystems. In this work, we introduce spatial regret, a graph-informed metric measuring the worst-case performance gap between a distributed controller and an oracle with access to additional sensor information. The oracle's graph is a user-specified augmentation of the information graph, yielding a benchmark policy that penalizes disturbances for which additional sensing would improve performance. Minimizing spatial regret yields distributed controllers—respecting the nominal information graph—that emulate the oracle's response to disturbances characteristic of large-scale networks, such as localized perturbations. We show that minimizing spatial regret admits a convex reformulation as an infinite program with a finite-dimensional approximation. To scale to large networks, we derive an upper bound on the spatial regret that can be efficiently minimized in a distributed way. Numerical experiments on power-system models show that the resulting controllers mitigate localized disturbances more effectively than those based on classical metrics.
\end{abstract}

\section{Introduction}
With the increasing scale and complexity of modern engineering systems, the traditional centralized control paradigm can become ineffective due to vulnerabilities such as single points of failure, high computational burden, and limited communication bandwidth~\cite{scattolini2009architectures}.
For these reasons, over the past few decades, significant effort has been devoted to the design of distributed control architectures, where local controllers are assigned to different subsystems and coordinate their actions through limited communication.
Typical applications include power systems~\cite{dorfler2014sparsity} and cooperative robotics~\cite{ren2008distributed}.

When using linear output-feedback controllers, complying with the communication network topology imposes structural constraints on the controller transfer function matrix in terms of sparsity patterns and fixed communication delays. Designing optimal distributed controllers under such structural constraints is notoriously challenging, as highlighted by Witsenhausen's counterexample~\cite{witsenhausen1968counterexample}.
A landmark contribution was given in~\cite{rotkowitz2005characterization}, where the authors proposed the Quadratic Invariance (QI), a sufficient~\cite{rotkowitz2005characterization} and necessary~\cite{lessard2011quadratic} condition for synthesizing sparse controllers that minimize closed-loop norms for linear systems \blue{in a convex way}.
More recently, it has been shown that QI is always satisfied for network controllers that share the same communication graph as the plant~\cite{vamsi2015optimal,rantzer2019realizability,naghnaeian2024youla}, while for sparsity patterns that do not satisfy QI, convex approximation methods have been developed~\cite{wang2019system,furieri2020sparsity}.\looseness-1

Beyond parameterizing distributed policies, a fundamental question is how to evaluate controller performance under disturbances.
Classical control methods rely on closed-loop norms such as $\H2$ and $\Hinfty$, which embed specific assumptions about disturbance characteristics. The $\H2$ norm assumes disturbances are generated by stochastic processes, whereas the $\Hinfty$ norm considers worst-case bounded-energy disturbances~\cite{goel2023regret}.
However, when real disturbances deviate from these assumptions, such metrics may fail to provide meaningful performance characterization.

This limitation is particularly pronounced in distributed control, where structural constraints on the policy further limit the ability to reject disturbances optimally. 
Hence, there is a need for alternative performance metrics that aim to quantify the fundamental trade-off between information availability and control performance. 
One approach is to assess how much improvement additional sensing or communication capabilities would bring to the distributed controller.
\blue{
For instance, \cite{wu2025predsls} introduces a scalable synthesis framework for distributed controllers that characterizes the trade-off between performance and the communication range of each local subcontroller.
}
Recent approaches, instead, have focused on synthesizing distributed controllers that approximate centralized policies. For example,~\cite{tolstaya2020learning,gama2022synthesizing} train neural networks to imitate centralized controllers, while \cite{fattahi2018transformation} gives conditions under which decentralized controllers closely approximate centralized ones for given initial conditions.
However, it is often infeasible to closely approximate the performance of a fully centralized controller.

\textbf{\textit{Contributions:}}
To inform classical performance metrics for distributed control with the ideal behavior reachable by a more informed benchmark, we propose \emph{spatial regret}. Spatial regret is a graph-informed metric that measures the worst-case performance difference between a distributed controller and a reference policy called the \textit{oracle}.
The oracle is a controller with additional sensing or communication links. Hence, its communication topology is a user-defined supergraph of the plant network, serving as a hypothetical (``what-if'') reference topology.
Compared to~\cite{tolstaya2020learning,gama2022synthesizing,fattahi2018transformation}, this topology does not need to be centralized, allowing for a trade-off between the performance improvement the oracle can achieve and how well we can mimic the oracle using only the available sensor measurements.
For example, suppose some nodes in a networked system are severely affected by disturbances. A suitable oracle topology could allow more distant nodes to access the measurements of these affected nodes, which helps reduce the spread of disturbances originating from the affected nodes. In this case, the spatial regret quantifies how much performance is lost due to the limited communication in the real network graph.
Thus, the oracle guides the design effort towards the parts of the system where extra information would have the largest benefit.

We analyze the theoretical properties of this new metric, establishing conditions for the oracle to guarantee that spatial regret is \emph{well-posed} in the sense that it is always non-negative. 
This ensures the oracle provides a meaningful performance benchmark.
We prove that spatial regret is well-posed as long as the oracle’s information structure contains that of the distributed controller, and the oracle is optimal in terms of $\H2$, $\Hinfty$, or $\mathcal{L}_1$ norm.
We then develop methods to synthesize distributed spatial regret-optimal controllers in an infinite-horizon setting. We derive a convex reformulation that casts the spatial regret minimization problem as an infinite-dimensional convex program, whose finite-dimensional approximation leads to a semidefinite program~(SDP).

\blue{Our approach differs from existing regret-optimal control methods~\cite{sabag2021regret, goel2023regret, martin2022safe, didier2022system, martin2023follow, hajar2024regret, martin2024guarantees, kargin2024wasserstein, yan2025distributional} in the nature of the information gap it captures: these works quantify a \emph{temporal} gap against a noncausal, centralized benchmark, whereas spatial regret quantifies a \emph{spatial} gap against a causal benchmark with enhanced sensing. Beyond this conceptual difference, existing formulations are not directly applicable to our setting. Riccati-based formulations~\cite{goel2023regret, sabag2021regret, hajar2024regret} cannot easily incorporate the sparsity constraints required for distributed policies, while convex optimization-based formulations~\cite{martin2022safe, didier2022system, martin2023follow, martin2024guarantees} handle such constraints naturally but have been limited to finite-horizon problems.
}
To address the computational scalability challenges of the resulting SDPs for large systems~(e.g., \cite{huang2022solving}), we develop an upper bound on spatial regret that measures the tracking error relative to oracle trajectories, thereby maintaining the oracle as a performance benchmark. The key advantage is that minimizing this upper bound leads to a linear program (LP) that can be solved in a distributed manner using the alternating direction method of multipliers (ADMM)~\cite{boyd2011distributed}, which decomposes the large problem into smaller subproblems that are coordinated to find the global solution.

Finally, we validate our framework through numerical experiments on power grid systems, demonstrating that spatial regret controllers outperform classical $\H2$ and $\Hinfty$ benchmarks under localized disturbances.
Using a 16-bus power grid model, we compare spatial regret-optimal controllers against traditional baselines across two scenarios: a 5-bus subsystem (where we obtain the optimal spatial regret controller) and the full 16-bus system (where we minimize the upper bound to the original metric).
The results show substantial performance improvements for multi-frequency disturbances affecting individual nodes, where the spatial regret framework effectively leverages the oracle's enhanced sensing capabilities.

This paper builds upon our conference version~\cite{martinelli2023closing}, where we first introduced the spatial regret framework for finite-horizon optimal distributed control with state feedback. 
The present work extends those results in three key directions. 
First, we transition from state to output feedback.
Second, we introduce more flexibility in the choice of oracle structure. In~\cite{martinelli2023closing}, the oracle communication graph was fixed to guarantee the well-posedness of the metric. 
In this work, we identify a class of networked plants and restrict the controller to share the plant's communication structure. Under these conditions, the spatial regret metric is well-posed for any oracle whose communication graph is a supergraph of the plant's.
This flexibility enables the incorporation of user insight into oracle design, broadening applicability to real-world networked systems.
Third, we extend the framework to infinite-horizon settings. The main challenge in this case is to derive tractable, finite-dimensional convex formulations that approximate the original infinite-dimensional programs while remaining dense, that is, capable of recovering the true solution as the approximation order increases.

\textbf{\textit{Paper organization:}}
Section~\ref{sec:problem_formulation} introduces the problem framework and discusses the challenges in designing optimal distributed controllers. Section~\ref{sec:main_results} presents the main theoretical results, including well-posedness conditions for the spatial regret metric and convex synthesis methods for optimal spatial regret controllers.
Section~\ref{sec:numerical_results} validates our approach through numerical experiments on a 16-bus power grid model, demonstrating superior performance against classical methods for localized disturbances. Section~\ref{sec:conclusion} concludes the paper and outlines future research directions.

\textbf{\textit{Notation:}}
The symbol $\real{a \times b}$ represents the set of real matrices with dimension $a \times b$. 
Let $\mathcal{R}_p^{a \times b}$ represent the set of real-rational proper transfer function matrices for a MIMO system with $a$ outputs and $b$ inputs.
$\mathcal{RH}_\infty^{a \times b}$ denotes the set of real-rational proper stable transfer function matrices of dimensions $a \times b$. 
We use the bold notation (e.g., $\tf{X}$) to denote transfer functions or, generally, Z-transforms of signals.
The notation $A(i,j)$ represents the element of matrix $A$ at position $(i,j)$. For any matrix \(A\) that is partitioned into blocks, we denote by \(A^{[i,j]}\) the \((i,j)\) block of \(A\), with subblock dimensions determined by context. We use similar notation for vectors and transfer function matrices: $x^{[i]}$ is the $i$-th subvector of $x$, and $\tf{P}^{[i,j]}$ denotes the $(i,j)$ block of $\tf{P}$.
With $\{x_{i,j}\}$ we denote the set of elements $x_{i,j}$ indexed by $i$ and $j$, and with $|\{x_{i,j}\}|$ its cardinality.
The symbols $A^*$ and $x^*$ denote the conjugate transpose of matrix $A$ and vector $x$, respectively.
For a given signal $x_t \in \real{n}$, $\norm{x}_2 = \sum_{t=0}^\infty \norm{x_t}^2$ is the squared $\ell_2$ norm, and $\norm{x}_\infty = \sup_{t \geq 0} \max_{i = 1, \dots, n} \abs{x_t(i)}$ is the $\ell_\infty$ norm.
The squared $\H2$ and $\Hinfty$ norms, and the $\mathcal{L}_1$ norm, of $\tf{P}(z) \in \mathcal{RH}_\infty^{m \times n}$ are defined as:
\preprintswitch{
$
    \norm{\tf{P}}_{\H2}^2  \coloneq \frac{1}{2\pi} \int_{-\pi}^{\pi} \operatorname{Trace} \left[ \tf{P}^*(e^{j\omega}) \tf{P}(e^{j\omega}) \right] \, d\omega 
$, $
    \norm{\tf{P}}_{\Hinfty}^2 \coloneq \sup_{\omega \in [-\pi, \pi]} \lambda_{\textnormal{max}}(\tf{P}^*(e^{j\omega}) \tf{P}(e^{j\omega}))
$, $
    \norm{\tf{P}}_{\mathcal{L}_1} \coloneq \max_{i = 1, \dots, m} \sum_{j=1}^{n} \sum_{t=0}^\infty \abs{P^{[i,j]}(t)}
$,}{
\begin{align*}
    \norm{\tf{P}}_{\H2}^2  &\coloneq \frac{1}{2\pi} \int_{-\pi}^{\pi} \operatorname{Trace} \left[ \tf{P}^*(e^{j\omega}) \tf{P}(e^{j\omega}) \right] \, d\omega \,, %
    \\
    \norm{\tf{P}}_{\Hinfty}^2 &\coloneq \sup_{\omega \in [-\pi, \pi]} \lambda_{\textnormal{max}}(\tf{P}^*(e^{j\omega}) \tf{P}(e^{j\omega}))\,, %
    \\
    \norm{\tf{P}}_{\mathcal{L}_1} &\coloneq \max_{i = 1, \dots, m} \sum_{j=1}^{n} \sum_{t=0}^\infty \abs{P^{[i,j]}(t)}\,, %
 \end{align*}
}
where $P^{[i,j]}(t)$ is the $(i,j)$ entry of the impulse response of $\tf{P}(z)$ at time $t$.
Consider a directed graph $\graph \coloneq (\vertices, \edges)$ with $N$ nodes, where $\vertices$ is the set of nodes and $\edges\subseteq \vertices^2$ is the set of directed edges.
The set $\edges$ contains $(i,j)$ if there exists a directed link from node $i$ to node $j$.
The directed edges of $\graph$ can be represented by the adjacency matrix $\adjmat$, where $\adjmat(j,i) = 1$ if $(i,j) \in \edges$, and $0$ otherwise.
The set of in-neighbors of node $i$ is $\mathcal{N}_{i} = \{j|(j,i) \in \edges\}$. 
We assume that each node always has a self-loop, i.e., $i \in \mathcal{N}_i \,\, \forall i \in \mathcal{V}$.

\section{Problem Formulation}\label{sec:problem_formulation}
\blue{
Let $\mathcal{G}$ be a directed graph with $N$ nodes. We consider discrete-time linear time-invariant (LTI) systems of the form
\begin{equation}\label{eq:LTI_system_P}
P:\quad
\begin{bmatrix} x_{t+1} \\ z_t \\ y_t \end{bmatrix}
=
\begin{bmatrix} A & B_1 & B_2 \\ C_1 & D_{11} & D_{12} \\ C_2 & D_{21} & D_{22} \end{bmatrix}
\begin{bmatrix} x_t \\ w_t \\ u_t \end{bmatrix},
\end{equation}
where $x_t \in \mathbb{R}^n$ is the state, $u_t \in \mathbb{R}^m$ the control input, $w_t \in \mathbb{R}^{n_w}$ the disturbance, $y_t \in \mathbb{R}^p$ the output, and $z_t \in \mathbb{R}^{n_z}$ the performance output. We are interested in distributed systems whose internal structure mirrors the topology of $\mathcal{G}$, as formalized next.
}

\blue{
\begin{definition}\label{def:network_structured}
The system $P$ in \eqref{eq:LTI_system_P} is \emph{network-structured on $\mathcal{G}$} if it admits a decomposition into $N$ subsystems whose dynamics, for every node $i \in \mathcal{V}$, follow:
{\setlength{\abovedisplayskip}{4pt}%
    \setlength{\belowdisplayskip}{4pt}%
\begin{equation}\label{eq:ss_networked_system}
\begin{aligned}
x^{[i]}_{t+1} &= \sum_{j \in \mathcal{N}_i} \left( A^{[i,j]} x^{[j]}_t + B^{[i,j]}_{1} w^{[j]}_t \right) + B^{[i,i]}_{2} u^{[i]}_t,\\
z^{[i]}_t &= \sum_{j \in \mathcal{N}_i} \left( C^{[i,j]}_{1} x^{[j]}_t + D^{[i,j]}_{11} w^{[j]}_t + D^{[i,j]}_{12} u^{[i]}_t \right),\\
y^{[i]}_t &= \sum_{j \in \mathcal{N}_i} \left( C^{[i,j]}_{2} x^{[j]}_t + D^{[i,j]}_{21} w^{[j]}_t \right) + D^{[i,i]}_{22} u^{[i]}_t.
\end{aligned}
\end{equation}
}
We denote by $\netstruc(\mathcal{G})$ the set of all such systems.
\end{definition}
}

\blue{
In \eqref{eq:ss_networked_system}, the local signals $x^{[i]}_t \in \mathbb{R}^{n_i}$, $u^{[i]}_t \in \mathbb{R}^{m_i}$, $w^{[i]}_t \in \mathbb{R}^{n_{w_i}}$, $y^{[i]}_t \in \mathbb{R}^{p_i}$, $z^{[i]}_t \in \mathbb{R}^{n_{z_i}}$ stack into the global signals as $\delta_t = \mathrm{vec}(\delta_t^{[1]}, \dots, \delta_t^{[N]})$ for $\delta \in \{x, u, w, y, z\}$, with $n = \sum_i n_i$ and similarly for $m$, $p$, $n_w$, $n_z$. Equivalently, $P \in \netstruc(\mathcal{G})$ means that $A$, $B_1$, $C_1$, $C_2$, $D_{11}$, $D_{21}$ admit a block partition whose $(i,j)$-block is zero whenever $j \notin \mathcal{N}_i$, while $B_2$, $D_{12}$, $D_{22}$ are block-diagonal, so each local input $u^{[i]}$ acts only on its own subsystem.
}
This structure is typical of distributed plants composed of interacting but locally controlled subsystems, such as power grids, multi-agent networks, and large-scale industrial processes. We can alternatively describe the dynamics of $P$ using its transfer function matrix $\tf{P}(\mathrm{z})$ defined by:
\begin{equation*}
    \begin{bmatrix}
        \tf{z} \\ \tf{y}
    \end{bmatrix}
    =
    \underbrace{\begin{bmatrix}
        \tf{P}_{11}(\textnormal{z}) & \tf{P}_{12}(\textnormal{z})
        \\
        \tf{P}_{21}(\textnormal{z}) & \tf{P}_{22}(\textnormal{z})
    \end{bmatrix}}_{\tf{P}(\operatorname{z})}
    \begin{bmatrix}
        \tf{w} \\ \tf{u}
    \end{bmatrix}
    \,,
\end{equation*}
where $\tf{P}_{ij}(\textnormal{z}) \coloneq C_{i} (\textnormal{z}I - A)^{-1}B_j + D_{ij}$.
Regarding system~\eqref{eq:ss_networked_system}, we make the following assumption.
\begin{assumption}\label{ass:networked_system_stable}
    The system \(P\) in~\eqref{eq:LTI_system_P} is network-structured on the graph \(\graph\) and is stable, in the sense that the matrix $A$ in~\eqref{eq:ss_networked_system} is Schur stable.
\end{assumption}
This assumption is made to simplify the presentation of the main contributions of the paper. For completeness, Appendix~\ref{app:review_when_plant_is_unstable} outlines the modifications to the synthesis framework when the plant is unstable.

In this work, we seek to design controllers that can be implemented in a distributed manner, ensure plant stability, and achieve optimality according to a chosen performance criterion.

Motivated by their optimality for closed-loop norm minimization~\cite{zhou1998essentials}, we focus on designing linear dynamical output-feedback controllers of the form $\tf{u} = \tf{K} \tf{y}$, with $\tf{K} \in \mathcal{R}_p^{m \times p}$, that can be implemented in a distributed manner. To this end, we focus on controllers that are network-structured on $\graph$, matching the plant communication topology. Hence, we require that $\tf{K} \in \netstruc(\graph)$.
This constraint ensures that $\tf{K}$ consists of $N$ subcontrollers, each communicating according to the same graph structure as the plant $\tf{P}$.
We illustrate this concept through the following example.
\begin{example}
Consider a network-structured plant $P$ over the graph $\graph$ with $\vertices = \{1 ,2 ,3 \}$ and  $\edges = \{(1, 2), (2, 1), (2, 3) \}$, controlled by a controller $\tf{K} \in \netstruc(\graph)$, as shown in Figure~\ref{fig:small_plant_scheme}. 
Each subcontroller $K^{[i]}$, $i \in \vertices$ evolves as follows:
{\setlength{\abovedisplayskip}{6pt}%
    \setlength{\belowdisplayskip}{3pt}%
\begin{equation} \label{eq:ss_dynamics_subcontroller}
\begin{aligned}
    \xi_{t+1}^{[i]} &= \sum_{j \in \mathcal{N}_i} \left(A_K^{[i,j]} \xi_t^{[j]} \right) + B_K^{[i,i]} y_t^{[i]}, \\
    u_t^{[i]} &= \sum_{j \in \mathcal{N}_i} \left(C_K^{[i,j]} \xi_t^{[j]} \right) + D_K^{[i,i]} y_t^{[i]},
\end{aligned}
\end{equation}}
where $\xi_t^{[i]} \in \real{n_{K_{i}}}$ is its internal state.
Equation~\eqref{eq:ss_dynamics_subcontroller} shows that $\tf{K}^{[i]}$ relies exclusively on locally available information: the measurement $y_t^{[i]}$ and the internal states $\xi_t^{[j]}$ of neighboring subcontrollers $j \in \mathcal{N}_i$.
\begin{figure}[t]
    \centering
    \includegraphics[trim= 20 10 20 10, clip, width=0.75\linewidth]{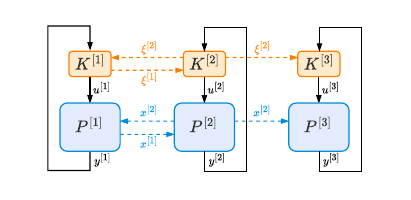}
    \caption{Scheme of an example of a plant $P \in \netstruc(\graph)$ controlled by $K \in \netstruc(\graph)$, where time dependencies are omitted for clarity.}
    \label{fig:small_plant_scheme}
\end{figure}
\end{example}

Another requirement for the controller $\tf{K}$ is that it must guarantee that the closed-loop is stable, i.e., $\tf{K}$ must belong to:
{\setlength{\abovedisplayskip}{0pt}%
    \setlength{\belowdisplayskip}{5pt}%
\begin{equation*}
    \mathcal{C}_{\textnormal{stab}} \coloneq \{\tf{K} \in \mathcal{R}_p^{m \times p} \mid \tf{K} \blue{\text{ internally stabilizes }  \tf{P} } \}.
\end{equation*}}
Finally, to evaluate controller performance, we adopt the following criterion: for $q\in \{2, \infty \}$, $x_0 = 0$, and a bounded disturbance $w$ with $\norm{w}_q < \infty$, we measure performance through ${\norm{{z}}_q}$.
Given this performance metric, our goal is to design distributed controllers that minimize it. We next recall classical results from optimal control theory~\cite{zhou1998essentials}.
Define the closed-loop transfer function matrix from $\tf{w}$ to $\tf{z}$:
{\setlength{\abovedisplayskip}{4pt}%
    \setlength{\belowdisplayskip}{4pt}%
\begin{equation}\label{eq:definition_LFT}
\tf{F}_{\ell}(\tf{P},\tf{K}) \coloneq \tf{P}_{11} + \tf{P}_{12} \tf{K}(I - \tf{P}_{22}\tf{K})^{-1} \tf{P}_{21}\,,
\end{equation}}
which is known as the \textit{lower fractional transformation}~(LFT) of $\tf{P}(\operatorname{z})$~\cite[Ch. 9.1]{zhou1998essentials}. 
We can then express the cost ${\norm{{z}}_q}$ as a function of $\tf{K}$:
{\setlength{\abovedisplayskip}{5pt}%
    \setlength{\belowdisplayskip}{5pt}%
\begin{equation}\label{eq:definition_cost_J}
J_q({w},\tf{K}) 
\coloneq
\norm{\mathcal{Z}^{-1}(\tf{z})}_q = \norm{\mathcal{Z}^{-1}(\tf{F}_{\ell}(\tf{P},\tf{K}) \tf{w}) }_q\,,
\end{equation}}
where $\mathcal{Z}^{-1}(\cdot)$ denotes the inverse Z-transform operation.
Ideally, the controller $\tf{K}$ should minimize $J_q(w,\tf{K})$ for any disturbance $w$. This, however, is impossible since $w$ is unmeasurable and unknown during the controller design phase~\cite{goel2023regret}.
In classical control theory, this challenge is addressed by instead minimizing $\mathbb{E}_{ w \sim \mathcal{N}(0 , I )} [J_2 (w, \tf{K}) ] $ and $\max_{\norm{w}_2 \leq 1} [J_2 (w, \tf{K})]$, corresponding respectively to the $\H2$ and $\Hinfty$ norms of $\tf{F}_{\ell}(\tf{P},\tf{K})$.
Hence, the classical optimal distributed control problem admits the following formulation:
\begin{equation} \label{eq:classical_optimal_distributed_control_problem}
    \min_{\tf{K}} \norm{\tf{F}_{\ell}(\tf{P},\tf{K})}\, \quad \st \quad \tf{K} \in \mathcal{C}_{\text{stab}} \cap \netstruc(\graph)\,,
\end{equation}
where $\norm{\cdot}$ is a chosen norm. 
The main advantage of using controllers that are network-structured on $\graph$ is that problem~\eqref{eq:classical_optimal_distributed_control_problem} can be solved in a convex way through the use of the Youla–Kučera parameterization. By combining results from~\cite{vamsi2015optimal,rantzer2019realizability,naghnaeian2024youla}, one can show that the set $\mathcal{C}_{\textnormal{stab}} \cap \netstruc(\graph)$ can be expressed in terms of the Youla parameter $\tf{Q}(\operatorname{z}) \coloneq \tf{K}(I - \tf{P}_{22}\tf{K})^{-1}$ as:
\begin{equation}\label{eq:definition_YK_parameterization}
    \mathcal{C}_{\textnormal{stab}} \cap \netstruc(\graph) = \{\tf{K} = \tf{Q}(I + \tf{P}_{22}\tf{Q})^{-1} \:| \: \tf{Q} \in \netstruc(\graph)_{\textnormal{stab}}\}\,,
\end{equation}
provided that $(I + D_{22}\tf{Q}(\infty))$ is invertible,\footnote{This condition is always satisfied when $D_{22}=0$, which holds for most physical systems. For simplicity, we assume it holds throughout the paper.} and where $\netstruc(\graph)_{\textnormal{stab}} \coloneq  \mathcal{RH}_\infty^{m \times p} \cap \netstruc(\graph)$.
This result enables us to rewrite problem~\eqref{eq:classical_optimal_distributed_control_problem} in a convex form as:
\begin{equation} \label{eq:classical_optimal_distr_problem_in_Q}
    \min_{\tf{Q}} \norm{\tf{F}_{\ell}(\tf{P},\tf{Q})}\, \quad \st \quad \tf{Q} \in \netstruc(\graph)_{\textnormal{stab}}\,,
\end{equation}
where  $\tf{F}_{\ell}(\tf{P},\tf{Q}) \coloneq \tf{P}_{11} + \tf{P}_{12} \tf{Q} \tf{P}_{21}$.
Appendix~\ref{app:review_when_plant_is_unstable} details the proof of~\eqref{eq:definition_YK_parameterization} and extends the result to unstable plants.
Hence, the performance of the optimal controller from~\eqref{eq:classical_optimal_distr_problem_in_Q} under disturbances depends on which closed-loop norm we minimize.
However, such classical metrics are inherently graph-agnostic: they treat all disturbances equally, irrespective of the network constraints imposed on the controller.

This motivates a different design philosophy: prioritize robustness against disturbances for which information limitations are the main bottleneck, rather than treating all disturbances uniformly. To do so, we consider an oracle controller $\tf{\hat{Q}} \in \netstruc(\hat{\graph})_{\textnormal{stab}}$ with access to an enhanced communication graph $\hat{\graph} \supset \graph$, where $\netstruc(\hat{\graph})_{\textnormal{stab}} \coloneq  \mathcal{RH}_\infty^{m \times p} \cap \netstruc(\hat{\graph})$.
We then define a new metric, termed \emph{spatial regret}, as:
\begin{equation}\label{eq:def_spatial_regret_cost}
        \spreg_q(\tf{Q},\tf{\hat{Q}}) 
        \coloneq
        \blue{\sup}_{\norm{{w}}_q \blue{= 1}}
        \left[
        J_q({w},\tf{Q}) - J_q({w},\tf{\hat{Q}}) \right]\,,
\end{equation}
where $J_q({w},\tf{Q}) \coloneq \norm{\mathcal{Z}^{-1}(\tf{F}_{\ell}(\tf{P},\tf{Q}) \tf{w}) }_q$, and the performance gap $\left[J_q(w,\tf{Q}) - J_q(w,\tf{\hat{Q}})\right]$ measures how much the information constraints cost for disturbance $w$. 
Ultimately, we seek to synthesize the distributed controller $\tf{Q}$ that minimizes spatial regret for $q \in \{2,\infty\}$ while preserving plant stability:
\begin{equation} \label{eq:min_problem_spregret_generic}
    \min_{\tf{Q}} \spreg_q(\tf{Q},\tf{\hat{Q}})\, \quad \st \quad \tf{Q} \in \netstruc(\graph)_{\textnormal{stab}}\,,
\end{equation}
Once solved, the optimal controller is reconstructed as $\tf{K}^* = \tf{Q}^* ( I + \tf{P}_{22} \tf{Q}^* )^{-1}$, satisfying $\tf{K}^* \in \mathcal{C}_{\textnormal{stab}} \cap \netstruc(\graph)$ as per~\eqref{eq:definition_YK_parameterization}.

This work addresses two main challenges. First, we establish conditions ensuring that the selected oracle provides a valid benchmark, so that the spatial regret metric is meaningful as a design objective. \blue{Second, we develop tractable convex methods to solve problem~\eqref{eq:min_problem_spregret_generic}: for $q=2$, an exact convex reformulation, and for $q=\infty$, a convex upper bound whose minimization admits a scalable distributed solution via ADMM. To address the first challenge,} we introduce the following definition.

\begin{definition}%
The spatial regret metric $\spreg_q(\tf{Q},\tf{\hat{Q}})$ is \emph{well-posed} for a given oracle $\tf{\hat{Q}} \in \netstruc(\hat{\graph})_{\textnormal{stab}}$ if
{\setlength{\abovedisplayskip}{3pt}%
    \setlength{\belowdisplayskip}{8pt}%
\begin{equation} \label{eq:def_wellposedness}
    \spreg_q(\tf{Q},\tf{\hat{Q}}) \geq 0\,, \quad \forall \tf{Q} \in \netstruc(\graph)_{\textnormal{stab}}\,.
\end{equation}}
\end{definition}
\vspace{0.5em}

Condition~\eqref{eq:def_wellposedness} guarantees that no distributed controller $\tf{Q}$ uniformly outperforms the oracle $\tf{\hat{Q}}$. When the condition is violated, the oracle ceases to be a meaningful benchmark, and the regret metric becomes misleading. The next scalar example illustrates how this failure can occur.

\begin{example}
Consider the following scalar LTI system:
{\setlength{\abovedisplayskip}{5pt}%
    \setlength{\belowdisplayskip}{5pt}%
\begin{equation}\label{eq:toy_example_wp}
x_{t+1} = \alpha x_t + \beta w_t + \gamma u_t,
\:\:\:
z_t = x_t,
\:\:\:
y_t = \begin{bmatrix}x_t & w_t\end{bmatrix}^\top \!\!\!,
\end{equation}}
where $\abs{\alpha} < 1$, $\beta \neq 0$, and $\gamma \neq 0$.
Since the system is already stable, one might naively choose the zero oracle $\tf{\hat{K}} = \tf{\hat{Q}} = 0$.
However, this choice violates well-posedness. 
To see why, consider that the static controller $ u_t=K^* y_t$, $K^* = -\frac{1}{\gamma}\begin{bmatrix}\alpha & \beta\end{bmatrix}$ cancels the closed-loop dynamics, yielding $x_{t+1} = 0$ for any $w$.
In contrast, the zero oracle incurs $J_q(w,\tf{\hat{Q}}) > 0$ for any nonzero $w$.
Hence, the oracle is uniformly outperformed by $\tf{Q}^* = K^*(I - \tf{P}_{22}K^*)^{-1}$, making it meaningless to mimic, and this is reflected in $\spreg_q(\tf{Q}^*, \tf{\hat{Q}}) < 0$.
\end{example}

\begin{remark}
Given $\tf{Q} \in \netstruc(\graph)_{\text{stab}}$ and $\tf{\hat{Q}} \in \netstruc(\hat{\graph})_{\text{stab}}$, computing the spatial regret cost \eqref{eq:def_spatial_regret_cost} is fundamentally different from computing the $\mathcal{H}_\infty$ norm, which corresponds to $ \max_{\norm{{w}}_2 \leq 1} \left[ J_2({w},\tf{Q}) \right]$.
The $\mathcal H_\infty$ norm is achieved by the disturbance $w$ that maximizes $J_2(w,\mathbf{Q})$.
By contrast, $\spreg(\tf{Q}, \tf{\hat{Q}} ) $ \blue{corresponds to the worst-case} $w$ \blue{maximizing the gap}
\(
J_2(w,\mathbf{Q}) - J_2(w,\hat{\mathbf{Q}}),
\)
thereby focusing on the disturbance providing the largest performance degradation compared to the benchmark $\tf{\hat{Q}}$.
\end{remark}

\begin{remark}
The spatial regret framework can naturally accommodate frequency-domain weights represented by the sparse transfer function matrices $\tf{W}(\mathrm{z}) \in \netstruc(\graph)$. By replacing $\tf{F}_{\ell}(\tf{P},\tf{K})$ with $\tf{F}_{\ell}(\tf{P},\tf{K})\tf{W}(\mathrm{z})$ \blue{in the spatial regret cost~\eqref{eq:def_spatial_regret_cost}}, it is possible to incorporate prior knowledge about disturbance characteristics, as usually done in $\mathcal{H}_2$ and $\mathcal{H}_\infty$ design~\cite{zhou1998essentials}.
\end{remark}

\section{Main results}\label{sec:main_results}
This section presents the main theoretical contributions. 
First, we establish conditions under which the oracle guarantees well-posedness of the spatial regret metric~\eqref{eq:def_spatial_regret_cost}. In the next theorem, we show that spatial regret is well-posed when the oracle’s information structure contains that of the distributed controller, and the oracle is optimal with respect to the $\H2$, $\Hinfty$, or $\mathcal{L}_1$ performance criteria.

\begin{theorem}\label{th:well_posedness}
Given a graph $\hat{\graph}$ such that 
$ \graph \subset \hat{\graph}$,
assume an oracle $\tf{\hat{Q}}^*$ is computed as the solution of:
{\setlength{\abovedisplayskip}{4pt}%
\begin{equation}\label{eq:def_oracle_problem_in_Q}
\min_{\tf{\hat{Q}}} f(\tf{\hat{Q}}) \;\;\; \st \quad \tf{\hat{Q}} \in \netstruc(\hat{\graph})_{\textnormal{stab}},
\end{equation}}
where:
\begin{equation*}
f(\tf{\hat{Q}}) = 
\begin{cases}
    \norm{\tf{F}_{\ell}(\tf{P},\tf{\hat{Q}})}_{\H2}^2 \: \text{or} \: \norm{\tf{F}_{\ell}(\tf{P},\tf{\hat{Q}})}_{\Hinfty}^2& \text{if } q = 2, \\[1ex]
    \norm{\tf{F}_{\ell}(\tf{P},\tf{\hat{Q}})}_{\mathcal{L}_1} & \text{if } q = \infty.
\end{cases}
\end{equation*}
Then, $\spreg_q (\tf{Q}, \tf{\hat{Q}}^*) \geq 0$ for any $\tf Q \in \netstruc(\graph)_{\textnormal{stab}}$.
\end{theorem}
The proof of Theorem~\ref{th:well_posedness} is reported in Appendix~\ref{app:proof_th_well_posedness_with_K}.

\blue{
Among the valid choices, the oracle graph $\hat{\mathcal{G}}$ trades off informativeness against attainability: if $\hat{\mathcal{G}} \approx \mathcal{G}$ the regret vanishes and carries no design information, whereas an overly rich $\hat{\mathcal{G}}$ exploits centralized information that no distributed controller can mimic, reducing spatial regret minimization to plain $\Hinfty$ ($q=2$) and $\mathcal{L}_1$ ($q=\infty$) norm minimization. A useful $\hat{\mathcal{G}}$ thus lies between, and in practice we recommend augmenting $\mathcal{G}$ with edges directed toward nodes identified a priori as disturbance-prone, such as renewable generation buses in electric networks (see also the examples in \Cref{sec:numerical_results}).
}

Having established the existence of well-posed oracle choices, the next step is to show how to solve problem~\eqref{eq:min_problem_spregret_generic} in a tractable manner. To this end, Section~\ref{subsec:convex_reformulation_SDP} derives an infinite-dimensional convex reformulation of problem~\eqref{eq:min_problem_spregret_generic} for $q=2$, along with a dense finite-dimensional approximation. Then, in Section~\ref{subsec:convex_reformulation_LP}, we present a convex upper bound for the case $q=\infty$, which additionally enables distributed optimization, making the approach scalable and thus more suitable for large-scale systems.

\subsection[Convex reformulation with q=2]{Convex reformulation of \(~\eqref{eq:min_problem_spregret_generic}\) with $q=2$}\label{subsec:convex_reformulation_SDP}
To proceed with the reformulation of problem~\eqref{eq:min_problem_spregret_generic}, we introduce the following technical result.
\begin{lemma}\label{lemma:Hinf_with_l2_signals}
    Let $\tf{Q}, \tf{\hat{Q}} \in \mathcal{RH}_{\infty}^{m \times p}$. It holds \blue{for $q=2$} that:
    \begin{equation} \label{eq:sup_lambda_max_Psi}
    \spreg_2(\tf{Q}, \tf{\hat{Q}}) = \sup_{\omega \in \Omega} \lambda_{\text{max}}
    ( \tf{\Psi} (e^{j\omega}) ),
    \end{equation}
    where:
\begin{align}
    \tf{\Psi} (e^{j\omega}) &\coloneq \tf{F}_\ell^*(e^{j\omega}) \tf{F}_\ell(e^{j\omega})-\tf{\hat{F}}_\ell^*(e^{j\omega}) \tf{\hat{F}}_\ell(e^{j\omega}) \,, 
\\
    \tf{F}_\ell(e^{j\omega}) &\coloneq \tf{P}_{11}(e^{j\omega}) + \tf{P}_{12}(e^{j\omega}) \tf{Q}(e^{j\omega}) \tf{P}_{21}(e^{j\omega})\,, \label{eq:def_F_ell_ejomega}
\\
    \tf{\hat{F}}_\ell(e^{j\omega}) &\coloneq \tf{P}_{11}(e^{j\omega}) + \tf{P}_{12}(e^{j\omega}) \tf{\hat{Q}}(e^{j\omega}) \tf{P}_{21}(e^{j\omega})\,, \label{eq:def_F_ell_hat_ejomega}
\end{align}
$\lambda_{\text{max}}(\tf{\Psi}(e^{j\omega}))$ denotes the largest eigenvalue of the matrix $\tf{\Psi}(e^{j\omega})$, and $\Omega \coloneq [-\pi, \pi]$ is the frequency spectrum.
\end{lemma}

The proof of~\Cref{lemma:Hinf_with_l2_signals} is reported in Appendix~\ref{app:proof_lemma_Hinf_with_l2_signals} and follows a structure similar to that of~\cite[Th. 4.3]{zhou1998essentials}. However, a direct application of the latter result is not possible for two reasons. First, \cite[Th. 4.3]{zhou1998essentials} is stated in continuous time, whereas our setting is discrete time. Second, the presence of the benchmark policy $\tf{\hat{Q}}$ in~\eqref{eq:def_spatial_regret_cost} leads to the difference of two quadratic forms, so the matrix function $\tf{\Psi}(e^{j\omega})$ may not be positive semidefinite for some $\omega \in \Omega$. As a consequence, a spectral factorization of the form $\tf{\Psi}(e^{j\omega}) = \ttf{\phi}^*(e^{j\omega}) \ttf{\phi}(e^{j\omega})$ (required to invoke \cite[Th. 4.3]{zhou1998essentials}) is not guaranteed to exist for every $\omega \in \Omega$.

Lemma~\ref{lemma:Hinf_with_l2_signals} demonstrates that the regret can be equivalently computed as the maximum eigenvalue of the function~$\tf{\Psi}(e^{j \omega})$ over the frequency spectrum, confirming that the cost $\spreg_2(\tf{Q}, \hat{\tf{Q}} )$ is convex in $\tf{Q}$. Indeed, $\tf{\Psi}(e^{j\omega})$ is a convex quadratic form in $\tf{Q}$, and the function $\lambda_{\text{max}}(\tf{\Psi}(e^{j \omega}))$ preserves convexity. 
The following theorem shows that problem~\eqref{eq:min_problem_spregret_generic} with $q=2$ can be reformulated as an infinite-dimensional convex program that admits a finite-dimensional approximation.
\looseness=-1

\begin{theorem}\label{th:convex_SDP_reformulation_spreg}
    Consider a networked plant ${P}$ that satisfies Assumption~\ref{ass:networked_system_stable}.
    Given a graph $\hat{\graph}$ such that $\graph \subset \hat{\graph}$, assume the oracle $\tf{\hat{Q}}$ is derived by solving the optimization problem~\eqref{eq:def_oracle_problem_in_Q}. %
Then the spatial-regret problem \eqref{eq:min_problem_spregret_generic} with $q=2$ is equivalent to the following convex program:
\begin{subequations}\label{eq:final_spre_min_prog_as_SDP}
\begin{alignat}{2}
    \min_{\lambda \,, \tf{Q}} ~& \lambda \label{eq:lambda_as_cost_convex_reformulation}
    \\
    \st~& 
    \tf{Q} \in \netstruc(\graph)_{\textnormal{stab}}\,,
    \quad \lambda > 0\,, \label{eq:constr_Q_netstruc_in_SDP_formulation}
        \\ &
    \begin{bmatrix}
    I & \tf{F}_\ell(e^{j \omega }) \\
    \tf{F}_\ell^*(e^{j \omega }) & \lambda I + \tf{\hat{F}}^*_\ell(e^{j \omega }) \tf{\hat{F}}_\ell(e^{j \omega })
    \end{bmatrix} \succeq 0, \:  \forall \omega \in \Omega\,, \label{eq:SDP_freqs_constraints}
\end{alignat}
\end{subequations}
where $\tf{F}_\ell(\textnormal{z})$ and $\tf{\hat{F}}_\ell(\textnormal{z})$ are defined in~\eqref{eq:def_F_ell_ejomega}, \eqref{eq:def_F_ell_hat_ejomega}. %
\end{theorem}

\begin{proof}
Notice that the cost function $\spreg(\tf{Q}, \tf{\hat{Q}})$ can be reformulated equivalently as~\eqref{eq:sup_lambda_max_Psi} thanks to~\Cref{lemma:Hinf_with_l2_signals}. Furthermore, introduce an auxiliary scalar $\lambda$ to bound the largest singular value of $\tf{\Psi}(e^{j\omega})$. This results in an equivalent epigraph formulation:
\begin{equation*}
    \min_{\lambda, \tf{Q}}  \lambda \quad  \st \:\:  \lambda > 0\,, \:\:
    \lambda I \succeq \tf{\Psi}(e^{j\omega}), \quad \forall \omega \in \Omega\,,
\end{equation*}
which gives~\eqref{eq:lambda_as_cost_convex_reformulation}-\eqref{eq:SDP_freqs_constraints} via Schur complement. Since all constraints are convex in $(\lambda , \tf{Q})$, the result follows.
\end{proof}

To approximate the infinite program~\eqref{eq:final_spre_min_prog_as_SDP}, we discretize the frequency-domain constraints~\eqref{eq:SDP_freqs_constraints} and adopt a finite parameterization of $\tf{Q} \in \netstruc(\graph)_{\textnormal{stab}}$.
Indeed, constraints~\eqref{eq:SDP_freqs_constraints} can be approximated by enforcing them over a finite set of frequencies $\Omega_\rho \coloneq \{\omega_1, \dots, \omega_\rho\} \subset \Omega$, with the flexibility to concentrate samples in frequency regions of interest (e.g., near resonance peaks).
We then derive a finite-dimensional representation of the constraint $\tf{Q} \in \netstruc(\graph)_{\text{stab}}$ in~\eqref{eq:constr_Q_netstruc_in_SDP_formulation}.
While Definition~\ref{def:network_structured} characterizes the network structure in state space, the constraint $\tf{Q} \in \netstruc(\graph)_{\text{stab}}$ requires a frequency-domain characterization, provided by~\cite{vamsi2015optimal}.
\begin{lemma}[\cite{vamsi2015optimal}]\label{lemma:network_structure_frequency}
    Given a system $Q$ such that $\tf{Q} \in \mathcal{R}_{p}^{m \times p}$, then $Q$ is network-structured on $\graph$ if and only if each block $\tf{Q}^{[i,j]}(\mathrm{z})$ satisfies
    \begin{equation} \label{eq:delay_structured_definition}
    \tf{Q}^{[i,j]}(\textnormal{z})
    =
        \begin{cases}
      \textnormal{z}^{-l(j,i)} \tf{H}_{i,j}(\textnormal{z}) &\textnormal{ if } l(j,i)<\infty, \\
      0 & \textnormal{ otherwise,}
        \end{cases}
    \end{equation}
    where $\tf{H}_{i,j}(\textnormal{z}) \in \mathcal{R}_p^{\blue{m_i \times p_j}}$ and $l(j,i)$ is the shortest path length from node $j$ to node $i$. If no path exists, we set $l(j,i) = \infty$; by convention, $l(i,i) = 0$.
\end{lemma}

Every $\tf{Q} \in \mathcal{RH}_{\infty}^{m \times p}$ admits the impulse response expansion:
{\setlength{\abovedisplayskip}{4pt}%
    \setlength{\belowdisplayskip}{4pt}%
\begin{equation}\label{eq:tf_as_sum_impulse_response}
    \tf{Q}(\textnormal{z}) 
    =
    \sum_{t = 0}^{\infty} {Q}_t \:\textnormal{z}^{-t}\,, \:\: Q_t \in \real{m \times p},
\end{equation}}
under which the condition $\tf{Q} \in \netstruc(\graph)$ becomes an infinite collection of linear constraints on the matrices $Q_t$.

\begin{proposition}\label{pr:Q_FIR_and_delay_struct}
    A transfer function $\tf{Q}(\textnormal{z})$ as in~\eqref{eq:tf_as_sum_impulse_response} satisfies $\tf{Q} \in \netstruc(\graph)_{\textnormal{stab}}$ if and only if $\forall i,j = 1, \dots, N$:
    \begin{equation}\label{eq:conditions_delay_structured_FIR}
        Q_{ijt} = 0\,, \quad \forall t : \:\: 0 \leq t < l(j,i)\,,
    \end{equation}
    where $Q_{ijt} \in \real{m_i \times p_j}$ is the $(i,j)$ block of $Q_t$.
\end{proposition}
This characterization supports a convex, finite-dimensional approximation of $\tf{Q} \in \netstruc(\graph)_{\textnormal{stab}}$. Truncating~\eqref{eq:tf_as_sum_impulse_response} to $f$ terms yields the finite impulse response~(FIR) model $\tf{\Tilde{Q}}(\textnormal{z}) = \sum_{t = 0}^{f} Q_t \:\textnormal{z}^{-t}$.
\blue{Since FIR models are inherently stable by construction ($\tf{\tilde{Q}} \in \mathcal{RH}_{\infty}^{m \times p}$), the closed-loop stability is guaranteed via \eqref{eq:definition_YK_parameterization}. Moreover,} \eqref{eq:tf_as_sum_impulse_response} reduces to a finite number of linear constraints on $Q_t$ for $t=0,\dots,f$.
Finally, problem ~\eqref{eq:final_spre_min_prog_as_SDP} can be approximated by solving the finite-dimensional SDP:
\looseness=-1

\vspace{-1em}
\begin{equation*}\label{eq:approx_final_spre_min_prog_as_SDP}
\begin{aligned}
   &\min_{ \lambda \,,\{Q_{ijt} \}} \quad  \lambda 
    \\
    &~~\st~
    \eqref{eq:conditions_delay_structured_FIR} \quad \forall i,j = 1, \dots, N\,, \forall t = 0, \dots, f\,,
    \\
    &~~~~~~~
    \begin{bmatrix}
    I & \tf{F}_\ell(e^{j \omega_k }) \\
    \tf{F}_\ell^*(e^{j \omega_k }) & \lambda I + \tf{\hat{F}}^*_\ell(e^{j \omega_k }) \tf{\hat{F}}_\ell(e^{j \omega_k })
    \end{bmatrix} \succeq 0, \:  \forall \omega_k \in \Omega_\rho\,,
\end{aligned}
\end{equation*}
where the optimization variables are the scalar $\lambda$ and the $(f+1)$ matrices $Q_{ijt} \in \real{m_i \times p_j}$ for $i,j = 1, \dots, N$.

\begin{remark} \label{remark:bis_method_goel}
\blue{
Problem~\eqref{eq:min_problem_spregret_generic} could in principle be solved using the bisection method from~\cite[Sec. IV]{goel2023regret}, which at each iteration solves a suboptimal $\Hinfty$ minimization problem. In~\cite{goel2023regret}, each subproblem is solved through a Riccati equation, which both yields the optimal controller and certifies that linear policies are optimal~\cite{zhou1998essentials}. In our case, however, the structural constraint $\tf{K} \in \netstruc(\graph)$ (or $\tf{Q} \in \netstruc(\graph)$) breaks both. Riccati-based synthesis is available only for suboptimal $\Hinfty$ under specific structures~\cite{lessard2014nested}, and the optimal case has been argued infeasible~\cite{scherer2013structured}; and even if it were available, certifying global linear optimality within distributed controllers is known only for $\H2$ under partially nested information~\cite{rotkowitz2008information}. One must then fall back on a convex treatment of the suboptimal $\Hinfty$ subproblem (e.g.,~\cite{alavian2013q,scherer2013structured}), where each iteration explicitly factorizes and inverts the oracle closed-loop transfer function matrix, an operation that scales poorly with the system size~(e.g.,~\cite{dey2023bit,jeannerod2005essentially}). Our reformulation in~\eqref{eq:final_spre_min_prog_as_SDP} instead enforces the required controller sparsity through the linear conditions~\eqref{eq:conditions_delay_structured_FIR}, at the price of sampling the frequency spectrum.
}
\end{remark}

While the $q=2$ case admits an exact SDP formulation, its computational complexity grows rapidly with system size (e.g., $\mathcal{O}(X^{6.5})$ for $X$ variables using interior-point methods~\cite{zhang2018sparse}), making it impractical for large-scale systems. 
To address this, we derive an upper bound on $\spreg_\infty$ whose minimization can be solved using ADMM, enabling scalability to large networks.
In contrast, the $q=2$ case cannot be distributed due to the frequency-domain constraints in~\eqref{eq:SDP_freqs_constraints} that couple all controller variables at each frequency.

\subsection[Convex reformulation with q=infty]{Convex reformulation of \(~\eqref{eq:min_problem_spregret_generic}\) with $q=\infty$}\label{subsec:convex_reformulation_LP}
Computing $\spreg_\infty$, as per~\eqref{eq:def_spatial_regret_cost}, requires evaluating the maximum difference of maximum absolute values over an infinite set of disturbances, which is computationally intractable. Unlike the $q=2$ case, there is no direct analogue of~\Cref{lemma:Hinf_with_l2_signals} that simplifies the formulation.
Instead, we address this computational challenge by deriving a convex upper bound that still extracts information from the dynamics of the oracle.

\begin{proposition}\label{prop:LP_reformulation_finite_dim}
Let $\tf{\hat{Q}}^*$ solve~\eqref{eq:def_oracle_problem_in_Q} with objective $f( \tf{\hat{Q}} ) = \|\tf{F}_\ell(\tf{P},\tf{\hat{Q}})\|_{\mathcal{L}_1}$. Then the optimal value of the problem:
\begin{equation}\label{eq:final_spreg_min_prog_as_LP}
\min_{\tf{Q}} \: \norm{\tf{F}_\ell(\tf{P},\tf{Q}) - \tf{F}_\ell(\tf{P},\tf{\hat{Q}}^*)}_{\mathcal{L}_1}
\:
\text{s.t.}
\:\:
\tf{Q} \in \netstruc(\graph)_{\textnormal{stab}},
\end{equation} 
provides an upper bound to~\eqref{eq:min_problem_spregret_generic} when $q=\infty$ in~\eqref{eq:def_spatial_regret_cost}. Moreover, problem~\eqref{eq:final_spreg_min_prog_as_LP} can be expressed as the following LP:
\begin{subequations}\label{eq:final_spreg_min_prog_as_LP_expanded}
\begin{align}
    &\min_{ \lambda, \{Q_{ijt}\}, \{\nu_{ijt}\}} \quad  \lambda \\
    &~\text{s.t.} \quad \eqref{eq:tf_as_sum_impulse_response}\,,\:\: \eqref{eq:conditions_delay_structured_FIR}\,,\quad 
        \forall (i,j) \in \mathcal{I}, \:\: t \geq 0:  \nonumber\\
        &~~~~ -\nu_{ijt} \leq \bar{P}_{11}^{[i,j]}(t) + ( P_{12} * Q * P_{21} )^{[i,j]}(t) \leq \nu_{ijt}\,, \label{eq:final_spreg_min_prog_as_LP_expanded_constr_convolution}
        \\
        &~~~~ \sum_{j=1}^{n_w} \sum_{t=0}^{\infty} \nu_{ijt} \leq \lambda\,, \quad \quad \nu_{ijt} \geq 0\,, \label{eq:final_spreg_min_prog_as_LP_expanded_constr_sums_over_nus}
\end{align}
\end{subequations}
where $\mathcal{I} \coloneq \{1,\dots,n_z\} \times \{1,\dots,n_w\}$, and \blue{$\nu_{ijt} \in \real{}$ for $(i,j) \in \mathcal{I}, \: t \geq 0$}. The terms $\bar{P}_{11}^{[i,j]}(t)$ and $( P_{12} * Q * P_{21} )^{[i,j]}(t)$ denote the $(i,j)^{th}$ impulse responses at time $t$ of $\tf{\bar{P}}_{11} \coloneq -\tf{P}_{12} \tf{\hat{Q}} \tf{P}_{21}$ and the convolution $\tf{P}_{12} * \tf{Q} * \tf{P}_{21}$, respectively.
\end{proposition}

\begin{proof}
Let $z$ and $\hat{z}$ denote the outputs of $\tf{F}_\ell(\tf{P},\tf{Q})$ and $\tf{F}_\ell(\tf{P},\tf{\hat{Q}})$, respectively, under a common disturbance $w$. 
By the triangle inequality, $\norm{z}_\infty - \norm{\hat{z}}_\infty \leq \norm{z - \hat{z}}_\infty$ for any $\norm{w}_\infty < \infty$ where both $\norm{z}_\infty$ and $\norm{\hat{z}}_\infty$ are finite due to the stability of $\tf{F}_\ell(\tf{P},\tf{Q})$ and $\tf{F}_\ell(\tf{P},\tf{\hat{Q}})$.
Then, taking the maximum over all disturbances $\norm{w}_\infty \blue{=} 1$ leads to:
\vspace{-5pt}
\begin{equation*}
    \underbrace{\blue{\sup}_{\norm{w}_\infty \blue{= 1}} \Big[ \norm{z}_\infty   - \norm{\hat{z}}_\infty \Big]}_{=\spreg_\infty(\tf{Q}, \tf{\hat{Q}})}
    \leq
    \underbrace{\blue{\sup}_{\norm{w}_\infty \blue{=} 1} \norm{z - \hat{z}}_\infty}_{ = \norm{\tf{F}_\ell(\tf{P},\tf{Q}) - \tf{F}_\ell(\tf{P},\tf{\hat{Q}})}_{\mathcal{L}_1}}.
\end{equation*}
This establishes the upper bound claim.
For the LP formulation, observe that $\tf{F}_\ell(\tf{P},\tf{Q}) - \tf{F}_\ell(\tf{P},\tf{\hat{Q}}) = \tf{\bar{P}}_{11} + \tf{P}_{12} \tf{Q} \tf{P}_{21}$ where $\tf{\bar{P}}_{11} \coloneq -\tf{P}_{12} \tf{\hat{Q}} \tf{P}_{21}$. Since both closed-loop maps are stable, their difference is stable, and the $\mathcal{L}_1$ norm of the difference can be expressed as:
{\setlength{\abovedisplayskip}{2pt}%
    \setlength{\belowdisplayskip}{3pt}%
\begin{multline}\label{eq:proof_l1_norm_difference_is_LP}
    \norm{\tf{F}_\ell(\tf{P},\tf{Q}) - \tf{F}_\ell(\tf{P},\tf{\hat{Q}})}_{\mathcal{L}_1} =
    \\
    \max_{i = 1, \dots , n_z} \sum_{j = 1}^{n_w} \sum_{t=0}^{\infty} \abs{  \bar{P}_{11}^{[i,j]}(t) +  ( P_{12} *Q * P_{21} )^{[i,j]}(t)}.
\end{multline}}
Standard epigraph reformulation of cost~\eqref{eq:proof_l1_norm_difference_is_LP} with slack variables $\{\nu_{ijt} \}$ leads to~\eqref{eq:final_spreg_min_prog_as_LP_expanded}, an LP in the variables $\{\nu_{ijt}\}$ and the impulse response coefficients $\{Q_{ijt}\}$. 
Finally, the constraint $\tf{Q} \in \netstruc(\graph)_{\text{stab}}$ is expressed as~\eqref{eq:conditions_delay_structured_FIR} thanks to~\Cref{pr:Q_FIR_and_delay_struct}. 
\end{proof}

Problem~\eqref{eq:final_spreg_min_prog_as_LP_expanded} is infinite due to the infinite number of time indices $t$ in constraints~\eqref{eq:conditions_delay_structured_FIR}, \eqref{eq:final_spreg_min_prog_as_LP_expanded_constr_convolution}, and \eqref{eq:final_spreg_min_prog_as_LP_expanded_constr_sums_over_nus}.
A finite-dimensional approximation is obtained by truncating the time horizon to $t = 0, \dots, f$ for some finite $f > 0$. This truncation is justified since both $\tf{Q} \in \netstruc(\graph)_{\textnormal{stab}}$ and the closed-loop difference $\tf{\bar{P}}_{11} + \tf{P}_{12} \tf{Q} \tf{P}_{21}$ are stable, which ensures that their impulse response coefficients decay exponentially to zero.

The upper bound $\norm{\tf{F}_\ell(\tf{P},\tf{Q}) - \tf{F}_\ell(\tf{P},\tf{\hat{Q}}^*)}_{\mathcal{L}_1}$ in \Cref{prop:LP_reformulation_finite_dim} penalizes the worst-case trajectory mismatch between the regulated outputs of $\tf{Q}$ and $\tf{\hat{Q}}$, rather than their worst-case difference in performance (as in $\spreg_\infty$). However, minimizing this trajectory mismatch still makes the approach graph-informed through the choice of the oracle supergraph $\hat{\graph}$.
\blue{
To quantify how conservative this upper bound is, in \preprintswitch{the Appendix of the extended manuscript~\cite{martinelli2025spatialregret}}{Appendix~\ref{app:lower_bound_spreg_infty}}, we derive a computable lower bound, denoted by $\spreg_\infty^{(T)}$, by only considering disturbances $w$ that vanish after $T$ time steps in the maximization of $\spreg_\infty$ in \eqref{eq:def_spatial_regret_cost}.
}

\subsection{Distributed computation} \label{subsec:distributed_computation_q_infty}
The LP in~\eqref{eq:final_spreg_min_prog_as_LP} has computational complexity scaling polynomially with $2(n_z n_w f) + 1$ variables, making it prohibitive for large-scale systems.
We address this by showing that problem~\eqref{eq:final_spreg_min_prog_as_LP} can be solved using distributed ADMM~\cite{boyd2011distributed,anderson2019system}, which decomposes the large problem into $M$ smaller subproblems solved in parallel, coordinated through dual variable updates.
To distribute problem~\eqref{eq:final_spreg_min_prog_as_LP}, we employ the system-level parametrization (SLP)~\cite{anderson2019system}. Given $n_r \coloneq n+m$ and $n_c \coloneq n + p$, introduce the auxiliary variable $\tf{\Phi} \in \mathcal{RH}_{\infty}^{n_r \times n_c}$:
\begin{equation}\label{eq:SLP_parametrization_original_form}
\begin{aligned}
    \tf{\Phi} &\coloneq
    \begin{bmatrix}
    \tf{\Phi}_{xx} & \tf{\Phi}_{xy} 
    \\
    \tf{\Phi}_{ux} & \tf{\Phi}_{uy}
\end{bmatrix}
\\
&\coloneq
 \begin{bmatrix}
     \left(\textnormal{z}I - A - B_2 \tf{K} C_2 \right)^{-1} & \tf{\Phi}_{xx} B_2 \tf{K}
     \\
     \tf{K} C_2 \tf{\Phi}_{xx} & \tf{K} ( I - \tf{P}_{22} \tf{K})^{-1}
 \end{bmatrix}\,.
\end{aligned}
\end{equation}
Then, $\tf{F}_{\ell} (\tf{P},\tf{K} )$ can be expressed in terms of $\tf{\Phi}$ as~\cite{anderson2019system}:
\begin{equation} \label{eq:SLP_parametrization_closed_loop_map}
\tf{F}_\ell(\tf{P},\tf{\Phi}) \coloneq \begin{bmatrix} C_1 & D_{12}\end{bmatrix} \tf{\Phi} \begin{bmatrix} B_1\\D_{21}\end{bmatrix} + D_{11}\,. 
\end{equation}
Under Assumption~\ref{ass:networked_system_stable}, we have $\tf{\Phi}_{uy} = \tf{Q}$, enabling direct imposition of network constraints on $\tf{\Phi}_{uy}$.
The difference $\tf{\Delta}(\tf{\Phi} , \tf{\hat{\Phi}}) \coloneq \tf{F}_\ell(\tf{P},\tf{\Phi}) - \tf{F}_\ell(\tf{P},\tf{\hat{\Phi}})$ becomes:
\begin{equation}\label{eq:SLP_parametrization_difference}
\tf{\Delta} ( \tf{\Phi} , \tf{\hat{\Phi}} )  = \begin{bmatrix} C_1 & D_{12} \end{bmatrix} (\tf{\Phi} - \tf{\hat{\Phi}}) \begin{bmatrix} B_1 \\ D_{21} \end{bmatrix},
\end{equation}
where $\tf{\hat{\Phi}}$ corresponds to the oracle $\tf{\hat{Q}}$. 
As shown in~\cite{anderson2019system}, the set of achievable closed-loop responses in~\eqref{eq:SLP_parametrization_original_form} is the set of all and only maps $\tf{\Phi}$ that satisfy:
\begin{subequations}\label{eq:achievability_constraints_SLP}
\begin{align}
\tf{\Phi} \begin{bmatrix} \textnormal{z}I - A \\ -C_2 \end{bmatrix} &= \begin{bmatrix} I \\ 0 \end{bmatrix}, \label{eq:achievability_constraints_SLP_row_wise} \\
\begin{bmatrix} \textnormal{z}I - A & -B_2 \end{bmatrix} \tf{\Phi} &= \begin{bmatrix} I & 0 \end{bmatrix}. \label{eq:achievability_constraints_SLP_column_wise}
\end{align}
\end{subequations}
Problem~\eqref{eq:final_spreg_min_prog_as_LP} can be equivalently rewritten in terms of $\tf{\Phi}$ as:
\begin{equation}\label{eq:SLP_problem}
\min_{\tf{\Phi}} \norm{\tf{\Delta}( \tf{\Phi} , \tf{\hat{\Phi}} )}_{\mathcal{L}_1} 
\:\:
\text{s.t.}
\:\:\:
\eqref{eq:achievability_constraints_SLP}, \:\: \tf{\Phi}_{uy} \in \netstruc(\graph)_{\textnormal{stab}}.
\end{equation}
Given the optimal $\tf{\Phi}^*$, the controller is recovered as $\tf{K}^* = \tf{\Phi}_{uy}^*(I + \tf{P}_{22}\tf{\Phi}_{uy}^*)^{-1} $.
In \cite{anderson2019system}, the authors propose an ADMM-based approach to solve classical optimal control problems in a distributed manner. Their method applies to problems with cost functions and constraints that depend separately on partitions of the optimization variable.
The constraints in~\eqref{eq:achievability_constraints_SLP} exhibit these required separability properties. Constraint~\eqref{eq:achievability_constraints_SLP_row_wise} is \emph{row-separable}, meaning that for any partition $r \coloneq \{r_1,\ldots,r_M\}$ of the row indices $\{1,\ldots,n_r\}$, it decomposes as:
\begin{equation*}
\tf{\Phi}^{[r_i,:]} \begin{bmatrix} \textnormal{z}I - A \\ -C_2 \end{bmatrix} = \begin{bmatrix} I \\ 0 \end{bmatrix}^{[r_i,:]}, \quad \forall i = 1,\ldots, M.
\end{equation*}
Similarly, constraint~\eqref{eq:achievability_constraints_SLP_column_wise} is \emph{column-separable} for any partition $c \coloneq \{c_1,\ldots,c_M\}$ of the column indices $\{1,\ldots,n_c\}$. The network constraint $\tf{\Phi}_{uy} \in \netstruc(\graph)_{\textnormal{stab}}$ imposes element-wise sparsity on impulse response coefficients~(\Cref{pr:Q_FIR_and_delay_struct}) and is both row- and column-separable.
We then partition the constraints of problem~\eqref{eq:SLP_problem} into row- and column-separable sets:
\looseness=-1
\begin{align}
    \mathcal{S}^{(r)} &\coloneq \left\{ \tf{\Phi} \in \mathcal{RH}_{\infty}^{n_r \times n_c} : \eqref{eq:achievability_constraints_SLP_row_wise}, \tf{\Phi}_{uy} \in \netstruc(\graph)_{\textnormal{stab}} \right\}, \label{eq:S_row_set} \\
    \mathcal{S}^{(c)} &\coloneq \left\{ \tf{\Phi} \in \mathcal{RH}_{\infty}^{n_r \times n_c} : \eqref{eq:achievability_constraints_SLP_column_wise} \right\}. \label{eq:S_col_set}
\end{align}
Similarly to constraint sets, a function $f(\tf{\Phi})$ is \textit{row-separable} for partition $\{r_1,\ldots,r_M\}$ of $\{1,\ldots,n_r\}$ if:
\begin{equation}\label{eq:definition_row_separable_cost_function}
    f(\tf{\Phi}) = \sum_{i=1}^{M} f^{(r)}_i(\tf{\Phi}^{[r_i, :]}).
\end{equation} 
For instance, the $\H2$ norm is row- and column-separable:
{\setlength{\abovedisplayskip}{4pt}%
    \setlength{\belowdisplayskip}{3pt}%
\begin{equation*}
    \norm{\tf{\Phi}}_{\H2}^2 = \sum_{i=1}^{M} \norm{\tf{\Phi}^{[r_i,:]}}_{\H2}^2  = \sum_{j=1}^{M} \norm{\tf{\Phi}^{[:,c_j]}}_{\H2}^2.
\end{equation*}}
However, the objective in~\eqref{eq:SLP_problem} is an $\mathcal{L}_1$ norm, which represents the maximum value in time over all outputs of a system. This norm couples all variables and prevents direct decomposition as in~\eqref{eq:definition_row_separable_cost_function}, causing our case to fall outside the standard framework of~\cite{anderson2019system}. 
Nevertheless, under structural assumptions on the regulated output matrices, problem~\eqref{eq:final_spreg_min_prog_as_LP} can still be solved in a distributed manner.

\begin{assumption}\label{ass:C1_D12_block_diagonal_permutation}
There exists a permutation matrix $\Pi \in \real{n_r \times n_r}$ such that $\begin{bmatrix} C_1 & D_{12} \end{bmatrix} \Pi$ is block-diagonal.
\end{assumption}
\blue{Since $C_1$ and $D_{12}$ are user-defined and do not depend on the plant, Assumption~\ref{ass:C1_D12_block_diagonal_permutation} is a mild condition on the choice of the performance objective rather than a structural restriction on the system, and is met in many practical cases.}
A representative example is the LQG cost $x^\top Q x + u^\top R u$, for which $C_1^\top = \begin{bmatrix} Q^{\top/2} & 0 \end{bmatrix}$ and $D_{12}^\top = \begin{bmatrix} 0 & R^{\top/2} \end{bmatrix}$ are  block-diagonal.

\begin{theorem}\label{thm:distributed_solution_ADMM}
Let $r \coloneq \{r_1,\ldots,r_M\}$ and $c \coloneq\{c_1,\ldots,c_M\}$ be row and column partitions of $\{1,\ldots,n_r\}$ and $\{1,\ldots,n_c\}$.
Define the functional $\mathcal{F}^{(r)}$ as:
{\setlength{\abovedisplayskip}{7pt}%
    \setlength{\belowdisplayskip}{4pt}%
\begin{equation}\label{eq:F_r_functional_definition}
    \begin{aligned}    
&\mathcal{F}^{(r)}(\gamma, \tf{\Phi}, \tf{\Psi}, \tf{\Lambda}, \tf{\hat{\Phi}}) \coloneq 
\\
&\begin{cases}
    \norm{\tf{\Phi} - (\tf{\Psi} - \tf{\Lambda})}_{\H2}^2 & \textnormal{if } \tf{\Phi} \in \mathcal{S}^{(r)}, \: \norm{\tf{\Delta}(\tf{\Phi}, \tf{\hat{\Phi}})}_{\mathcal{L}_1} \leq \gamma,
    \\
    +\infty & \textnormal{otherwise,}
\end{cases}
\end{aligned} 
\end{equation}}
\blue{and} $\tf{\Delta}_{r_i}$ collects the components of $\tf{\Delta}(\tf{\Phi}, \tf{\hat{\Phi}})$ that depend only on $\tf{\Phi}^{[r_i,:]}$ (as induced by the permutation matrix $\Pi$).
Moreover, define for each $j \in \{1,\ldots,M\}$ the following functionals:
{\setlength{\abovedisplayskip}{4pt}%
    \setlength{\belowdisplayskip}{4pt}%
\begin{equation}\label{eq:F_c_functional_definition}
    \mathcal{F}^{(c)}_j(\tf{\Psi}_j, \tf{\Phi}, \tf{\Lambda} ) \coloneq 
    \begin{cases}
        \norm{\tf{\Psi}_j - (\tf{\Phi} + \tf{\Lambda})^{[:,c_j]}}_{\H2}^2 & \text{if } \tf{\Psi}_j \in \mathcal{S}^{(c)}_j, \\
        +\infty & \text{otherwise,}
    \end{cases}
\end{equation}}
where $\tf{\Psi}_j \in \mathcal{RH}_\infty^{n_r \times |c_j|}$ and $\mathcal{S}^{(c)}_j$ represents the restriction of the column constraint set $\mathcal{S}^{(c)}$ to the variables $\tf{\Psi}^{[:,c_j]}$.
Under Assumption~\ref{ass:C1_D12_block_diagonal_permutation} with $\rho > 0$, problem~\eqref{eq:SLP_problem} can be solved using the distributed Algorithm~\ref{alg:ADMM_spreg}, which converges to consensus $\tf{\Phi}^* = \tf{\Psi}^*$, $\tf{\Lambda}^* = 0$, with $\tf{\Phi}^*$ minimizing~\eqref{eq:SLP_problem}. 
\end{theorem}

\preprintswitch{
For the sake of conciseness, the proof of~\Cref{thm:distributed_solution_ADMM} is deferred to the Appendix of the extended manuscript~\cite{martinelli2025spatialregret}.
}
{
\begin{proof}
The proof is reported in Appendix~\ref{app:proof_theorem_distributed_solution_ADMM}.
\end{proof}
}

Algorithm~\ref{alg:ADMM_spreg} enables distributed solution of~\eqref{eq:final_spreg_min_prog_as_LP}: steps~\ref{alg_step:gamma_minimization_golden_search}-\ref{alg_step:phi_k_argmin_minimization} minimize $\mathcal{F}^{(r)}$ in parallel via Algorithm~\ref{alg:parallel_min_F(r)}, while steps~\ref{alg_step:for_loop_col}-\ref{alg_step:column_wise_auxiliary_update} solve $M$ smaller subproblems in parallel.
However, this computational advantage comes with a trade-off: at each iteration $k$, step~\ref{alg_step:gamma_minimization_golden_search} involves a line search (e.g., the bisection method) to determine the optimal $\gamma_{k+1}$, requiring $\mathcal{O}(\log(\epsilon^{-1}))$ iterations to achieve accuracy $\epsilon$.

\begin{algorithm}[t]
\caption{Distributed optimization of~\blue{\eqref{eq:final_spreg_min_prog_as_LP}} with $q=\infty$.}
\label{alg:ADMM_spreg}
\centering
\begin{algorithmic}[1]
\algrenewcommand{\algorithmicindent}{0em}
\Require Oracle $\tf{\hat{\Phi}}$, $\rho > 0$, partitions $r = \{r_1, \ldots, r_M\}$, $c = \{c_1, \ldots, c_M\}$
\State Initialize: $\tf{\Psi}_0$, $\tf{\Lambda}_0$
\State \textbf{for} $k = 0, 1, 2, \ldots$ \textbf{do:}
    \State \hspace{0.5em} $\gamma_{k+1} \gets \arg\min_{\gamma} \left[ \gamma + \frac{\rho}{2}\min_{\tf{\Phi}}\mathcal{F}^{(r)}(\gamma , \tf{\Phi}, \tf{\Psi}_k, \tf{\Lambda}_k, \tf{\hat{\Phi}}) \right]$ \label{alg_step:gamma_minimization_golden_search}
    \vspace{-0.7em}
    \State \hspace{0.5em} $\tf{\Phi}_{k+1} \gets  \arg\min_{\tf{\Phi}}\mathcal{F}^{(r)}(\gamma_{k+1} , \tf{\Phi}, \tf{\Psi}_k, \tf{\Lambda}_k, \tf{\hat{\Phi}})$ \label{alg_step:phi_k_argmin_minimization}
    \State \hspace{0.5em} \textbf{for} $j = 1, \ldots, M$ \textbf{in parallel do:} \label{alg_step:for_loop_col}
        \State \hspace{1em} $\tf{\Psi}_{k+1}^{[:,c_j]} \gets \arg\min_{\tf{\Psi}_j} \mathcal{F}^{(c)}_j(\tf{\Psi}_j, \tf{\Phi}_{k+1}, \tf{\Lambda}_k)$ \label{alg_step:column_wise_auxiliary_update}
    \State \hspace{0.5em} $\tf{\Lambda}_{k+1} \gets \tf{\Lambda}_k + (\tf{\Phi}_{k+1} - \tf{\Psi}_{k+1})$ \label{alg_step:dual_variable_update}
\Ensure Optimal controller parameter $\tf{\Phi}^*$
\end{algorithmic}
\end{algorithm}

\begin{algorithm}[t]
\caption{Parallel row-wise minimization of $\mathcal{F}^{(r)}$}
\label{alg:parallel_min_F(r)}
\centering
\begin{algorithmic}[1]
\algrenewcommand{\algorithmicindent}{0em}
\Require $\gamma$, $\tf{\Psi}_k$, $\tf{\Lambda}_k$, $\tf{\hat{\Phi}}$, partition $r = \{r_1, \dots, r_M \} $
\State Initialize: $\text{sum} \gets 0$
\State \textbf{for} $i = 1, \dots, M$ \textbf{in parallel do:}
    \State \hspace{0.5em} \(
    \text{sum} \gets \text{sum} + 
    \left(
    \begin{aligned}
    \min_{\tf{\Phi}_{i}\in \mathcal{S}^{(r_i)}}&  \norm{\tf{\Phi}_i - (\tf{\Psi} - \tf{\Lambda})^{[r_i,:]}}_{\H2}^2 \\
\st~&  \norm{\tf{\Delta}_{r_i}}_{\mathcal{L}_1} \leq \gamma
    \end{aligned}
    \right)
    \)
\Ensure $\text{sum} = \min_{\tf{\Phi}}\mathcal{F}^{(r)}(\gamma , \tf{\Phi}, \tf{\Psi}_k, \tf{\Lambda}_k, \tf{\hat{\Phi}})$
\end{algorithmic}
\end{algorithm}

\section{Numerical results} \label{sec:numerical_results}

\begin{figure}[t]
    \centering
    \includegraphics[trim= 15 10 20 10, clip, width=0.52\linewidth]{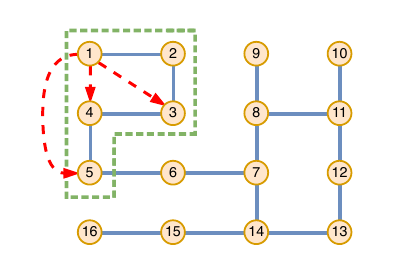}
    \caption{Graph of the 16-bus power system model. Blue lines show the undirected edges between buses. Red dashed arrows indicate oracle's additional connections. Green dashed lines highlight the 5-bus subsystem considered for the first example.
    }
    \label{fig:grid_scheme}
\end{figure}

This section validates the spatial regret framework through numerical studies on power systems for $q\in \{2, \infty \}$, highlighting how graph-informed design enhances disturbance rejection and scales with system size.\footnote{The code used in this section is available at~\url{https://github.com/DecodEPFL/SpatialRegret}.}

We analyze the swing dynamics of a 16-bus power grid following the model in~\cite{anderson2019system}, with network topology shown in Figure~\ref{fig:grid_scheme}. Due to the scalability limitations of the SDP-based approach~\eqref{eq:final_spre_min_prog_as_SDP}, we initially control only the first 5 buses (highlighted in green in Figure~\ref{fig:grid_scheme}) and subsequently consider the full 16-bus system.
Each bus $i \in \{1,\ldots,16\}$ is modeled as a dynamical subsystem with state ${x}^{[i]}_{t} \in \mathbb{R}^2$ representing phase angle and frequency deviations. For simplicity, all buses share identical parameters: inertia $m_i = \SI{1}{}$, damping $d_i = \SI{2}{}$, and coupling strength $k_{ij} = \SI{20}{}$. These parameters follow~\cite{anderson2019system}, except for a higher coupling gain $k_{ij}$ to make the system more challenging to control. All quantities are expressed in per-unit values.
Each subsystem evolves according to~\eqref{eq:ss_networked_system} with matrices:
\begin{gather*}
    A^{[i,i]} \!\!= \!\!\begin{bmatrix}
        1 & T_s \\
        -\frac{k_i}{M}T_s & 1 - \frac{d_i}{M}T_s
    \end{bmatrix}\!, \:\:
        A^{[i,j]} \!\!= \!\! \begin{bmatrix}
        0 & 0 \\
        \frac{k_{ij}}{M}T_s & 0
    \end{bmatrix}\!\!,\:\:
        B_{1}^{[i,i]} \!\!= \!\! \begin{bmatrix}
        0\\
        1
    \end{bmatrix}
    \\
        B_{2}^{[i,i]} \!= \! \begin{bmatrix}
            0 \\ \frac{T_s}{M}
        \end{bmatrix},
    \:\:
        C_{1}^{[i,i]} \!= \!  \begin{bmatrix}
            1 \\ 0
        \end{bmatrix},
    \:\:
        C_{2}^{[i,i]} \!= \!  \begin{bmatrix}
            1 \\ 0
        \end{bmatrix}^\top \!\!\!, 
    \:\:
        D_{12}^{[i,i]} \!= \!  \begin{bmatrix}
            0 \\ 1
        \end{bmatrix},
    \\
    D_{12}^{[i,j]} = 0 \text{ for } i \neq j, 
    \:
    \:
        D_{21}^{[i,i]} = 1,  
    \:
    D_{11} = D_{22} = 0\,,
\end{gather*}  
where $T_s = \SI{0.1}{s}$ is the sampling time.
Because the plant is open-loop unstable (with two poles in $0.9 \pm 0.956j$), we pre-stabilize it using the method described in~\Cref{remark:pre_stabilization_with_F_and_L}.

\subsection[Distributed control of 5 buses using SpReg-2]{Distributed control of 5 buses using $\spreg_2$}

We focus our attention on the first 5 buses of the network.
Our framework aims to develop a controller that, by mimicking a graph-informed oracle, achieves better performance compared to traditional control techniques for disturbances that fall outside the classical assumptions.
To investigate this, we assume that the first bus is prone to disturbances. Hence, we synthesize an oracle $\tf{\hat{Q}}$ in which each subcontroller has immediate access to the $1^{st}$ bus controller states. In Figure~\ref{fig:grid_scheme}, the additional connections in the oracle graph are shown in red. Then, we compute the oracle $\tf{\hat{Q}}$ by solving problem~\eqref{eq:def_oracle_problem_in_Q} 
where $f(\tf{\hat{Q}}) = \norm{\tf{F}_\ell(\tf{P}, \tf{\hat{Q}})}_{\Hinfty}^2$.
\blue{We then obtain an approximated $\spreg_2$ optimal controller $\tf{K}^{\textnormal{SR}}$ by solving problem~\eqref{eq:final_spre_min_prog_as_SDP}, where $\Omega_\rho$ is a uniform grid of $\rho = 4000$ points over $\Omega$}.
For comparison, we compute two benchmarks $\tf{K}^{\H2}$ and $\tf{K}^{\Hinfty}$, the optimal $\H2$ and $\Hinfty$ distributed controllers:
\begin{align}
    \tf{K}^{\H2} = \arg\min_{\tf{K} \in \mathcal{C}_{\text{stab}} \cap \netstruc(\graph)} &\norm{\tf{F}_\ell(\tf{P}, \tf{K})}_{\H2}^2 \label{eq:H2_problem}\,,
    \\
    \tf{K}^{\Hinfty} = \arg\min_{\tf{K} \in \mathcal{C}_{\text{stab}} \cap \netstruc(\graph)} &\norm{\tf{F}_\ell(\tf{P}, \tf{K})}_{\Hinfty}^2 \label{eq:Hinfty_problem} \,,
\end{align}
using the methods in~\cite{scherer2000efficient,alavian2013q}, where \eqref{eq:H2_problem} and \eqref{eq:Hinfty_problem} are rewritten as SDPs, under the assumption that the plant $\tf{P}_{22}$ has been pre-stabilized as in~\Cref{remark:pre_stabilization_with_F_and_L}, and using the FIR approximation of the Youla parameter described in~\Cref{subsec:convex_reformulation_SDP} \blue{with $f=20$}.
For each controller, we compute the $\H2$-norm, $\Hinfty$-norm, and the $\spreg_2$-metric of the resulting closed-loop system.
We report the results in the table below.
\vspace*{-0.3cm}
\begin{table}[H]
    \centering
    \renewcommand{\arraystretch}{1.3}
    \setlength{\tabcolsep}{2pt}
    \footnotesize
    \rowcolors{1}{}{lightgrayshblue}
    \begin{tabular}{c|c|c|c}
        Controller & $\norm{\tf{F}_{\ell}(\tf{P}, \: \cdot \:)}_{\H2}^2$ & $\norm{\tf{F}_{\ell}(\tf{P}, \: \cdot \:)}_{\Hinfty}^2$ & $\spreg_2( \: \cdot \: , \tf{\hat{K}})$ \\
        \hline
        $\tf{K}^{\H2}$   &   \textbf{875.2}    &  2084.3  & 1066.2 \\ \hline
        $\tf{K}^{\textnormal{SR}}$   & 1598.2    &  1050.8 & \textbf{33.89}  \\ \hline
        $\tf{K}^{\Hinfty}$    &  1702.5    &  \textbf{1020.0}  & 115.04 \\
        \hline
    \end{tabular}
\end{table}
\vspace*{-0.3cm}
As expected, the controller $\tf{K}^{\text{SR}}$ achieves the lowest spatial regret. 
However, these results consider only disturbances applied to all subsystems simultaneously.
Due to the design of the oracle, $\tf{K}^{\text{SR}}$ is expected to better reject disturbances locally affecting bus 1.
To investigate this, we perform a spectral analysis by applying disturbances of the form:
\preprintswitch %
{$\bar{w}_t \coloneq \begin{bmatrix}w^{[1]}_t & 0 & \cdots & 0\end{bmatrix},$}
{\begin{equation*}\bar{w}_t \coloneq \begin{bmatrix}w^{[1]}_t & 0 & \cdots & 0\end{bmatrix}\,, \end{equation*}}
where $w^{[1]}_t = \cos(\omega t)$, and we denote with $\bar{z}_t$ its corresponding measured output. Hence, for each controller and for each frequency $\omega \in \Omega$, we compute $\norm{\tf{F}_{\ell}^{[:,1]}(e^{j \omega})}_2^2$, which satisfies:
\begin{equation*}
    \norm{\tf{F}_{\ell}^{[:,1]}(e^{j \omega})}_2^2 =     \begin{cases}
        \langle \norm{\bar{z}}^2 \rangle & \textnormal{if } \omega = 0 \textnormal{ or } \omega = \pm \pi
        \\
        2\:\langle \norm{\bar{z}}^2 \rangle & \textnormal{otherwise}\,,
    \end{cases}
\end{equation*}
where $\langle \norm{\bar{z}}^2 \rangle \coloneq \lim_{T\to \infty} \frac{1}{T} \sum_{t=0}^{T-1} \norm{\bar{z}_t}^2$ is the time-averaged energy of the regulated output when $w_t = \bar{w}_t$.
Figure~\ref{fig:plot_with_5_masses_on_5_channel} plots $\norm{\tf{F}_{\ell}^{[:,1]}(e^{j \omega})}_2^2$ as a function of frequency.
Here, the spatial regret controller $\tf{K}^{\textnormal{SR}}$ closely mimics the oracle controller $\tf{\hat{K}}$, leading to improved performance compared to both $\tf{K}^{\H2}$ (for the peak around $\omega\approx \frac{\pi}{4}$) and $\tf{K}^{\Hinfty}$ (for frequencies $ \omega < \frac{\pi}{8}$ and $\omega > \frac{\pi}{2}$).
To validate these predictions, we provide a simulation using a multi-frequency disturbance with $w^{[1]}_t =\cos(0.1 t) + \cos (\frac{\pi}{5} t)$. %
\preprintswitch{}{
Figure~\ref{fig:response_5_masses_sinusoidal_disturbance} presents the results in the time domain.
}
The spatial regret controller achieves a $27.12\%$ reduction in average $\norm{z}_2$ compared to $\tf{K}^{\mathcal{H}_2}$ and $14.75\%$ compared to $\tf{K}^{\mathcal{H}_\infty}$.

\begin{figure}[t]
    \centering
    \includegraphics[ width=0.35\textwidth]{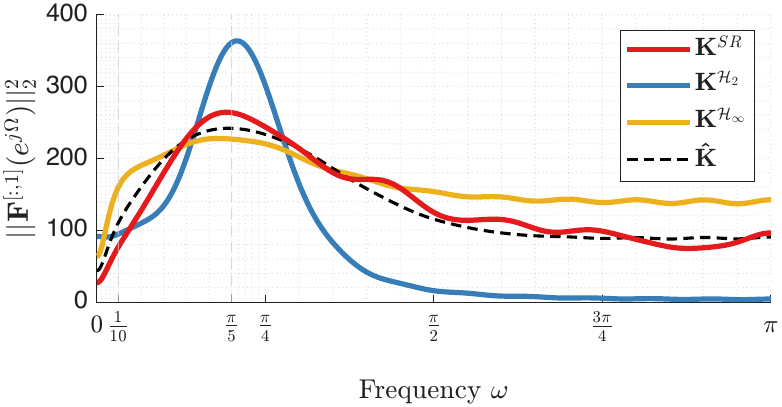}
    \caption{Plot of the squared 2-norm of $\tf{F}_{\ell}^{[:,1]}(e^{j \omega})$ as a function of frequency.}
    \label{fig:plot_with_5_masses_on_5_channel}
\end{figure}
\preprintswitch{}{
\begin{figure}[t]
    \centering
    \includegraphics[width=0.35\textwidth]{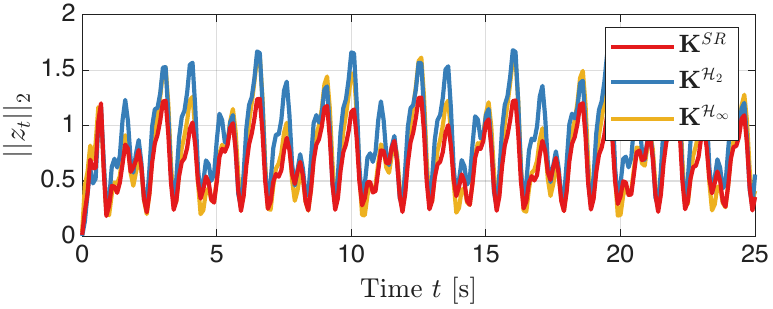}
    \caption{Plot of $\norm{z_t}_2$ as a function of time for the three controllers with $w_t = \bar{w}_t$ and $w^{[1]}_t =\cos(0.1 t) + \cos (\frac{\pi}{5}t)$.}
    \label{fig:response_5_masses_sinusoidal_disturbance}
\end{figure}
}

\subsection[Distributed control of 16 buses using SpReg-Infty]{Distributed control of 16 buses using $\spreg_\infty$}

We extend the analysis to the full 16-bus system, where the SDP approach~\eqref{eq:final_spre_min_prog_as_SDP} is intractable at this scale. Instead, we use the more scalable distributed method from~\Cref{subsec:distributed_computation_q_infty}.

The $\mathcal{L}_1$-optimal oracle controller $\tf{\hat{K}}$ is constructed as in the first case study, with subcontrollers $3$-$5$ given direct access to bus $1$ subcontroller states, while the rest remains unchanged. 
$\tf{\hat{K}}$ is computed via Algorithm~\ref{alg:ADMM_spreg} without specifying an oracle and using the graph $\hat{\graph}$.
The optimal controller $\tf{K}^{\text{SR}}_{\mathcal{L}_1}$ for problem~\eqref{eq:final_spreg_min_prog_as_LP} is then synthesized by running Algorithm~\ref{alg:ADMM_spreg}. As a benchmark, we consider $\tf{K}^{\mathcal{L}_1}$, the optimal $\mathcal{L}_1$ controller network-structured on $\graph$ (synthesized similarly to $\tf{\hat{K}}$).

Numerical experiments confirm the same trends as in the $q=2$ case, despite considering $\ell_\infty$-bounded disturbances and optimizing the upper bound~\eqref{eq:final_spreg_min_prog_as_LP}.
The nominal $\mathcal{L}_1$ controller achieves a slightly lower $\mathcal{L}_1$ norm ($83.24$) than the spatial regret controller ($83.30$), which considers only disturbances affecting all subsystems.
However, this gap is compensated by improved performance for spatially localized disturbances. For each frequency $\omega \in \Omega$, we compute:
\preprintswitch{
    $
\norm{\tf{F}_\ell^{[:,1]}(e^{j \omega})}_{\infty}
=
\max_{i=1,\dots,n_z} \left|\tf{F}_\ell^{[i,1]}(e^{j \omega})\right|
=
\norm{\bar{z}}_{\infty}
$
}{
   \begin{equation*}
    \norm{\tf{F}_\ell^{[:,1]}(e^{j \omega})}_{\infty}
=
\max_{i=1,\dots,n_z} \left|\tf{F}_\ell^{[i,1]}(e^{j \omega})\right|
=
\norm{\bar{z}}_{\infty}
   \end{equation*} 
}
which captures the peak amplification across all output channels of $\bar{z}_t$ when the input disturbance is $w_t = \bar{w}_t$.
Figure~\ref{fig:both_plots_with_10_masses} shows $\norm{\tf{F}_\ell^{[:,1]}(e^{j \omega})}_{\infty}$ as a function of frequency. The spatial regret controller $\tf{K}^{SR}_{\mathcal{L}_1}$ closely matches the oracle $\tf{\hat{K}}$, outperforming $\tf{{K}}^{\mathcal{L}_1}$ for most frequencies, especially at high frequencies ($\omega > \frac{3\pi}{4}$). Integrating over all frequencies, the spatial regret controller achieves a $3.12\%$ relative improvement compared to the benchmark.

To further validate robustness, we simulate both controllers under a disturbance composed of a sum of cosines with 100 frequencies in $\left[ \frac{3\pi}{4}, \pi \right]$ by imposing $w^{[1]}_t = \sum_{k=0}^{100}\cos\left(\frac{1}{100}\left(\frac{3\pi}{4}(100-k) + \pi k \right) t + \phi_k\right)$, 
where $\phi_k$ is drawn uniformly from $[0, 2\pi)$. %
\preprintswitch{}{
Figure~\ref{fig:response_10_masses_sinusoidal_disturbance} shows the resulting $\|z\|_{\infty}$ evolution for one realization.
}
Over $10^4$ random realizations, $\tf{K}^{SR}_{\mathcal{L}_1}$ achieves an average $9.70\%$ reduction in peak output magnitude compared to the baseline.

In summary, the spatial regret controller leverages additional graph information to improve robustness to localized disturbances in large-scale systems, while maintaining comparable nominal $\mathcal{L}_1$ performance.

\begin{figure}[t]
    \centering
    \includegraphics[width=0.35\textwidth]{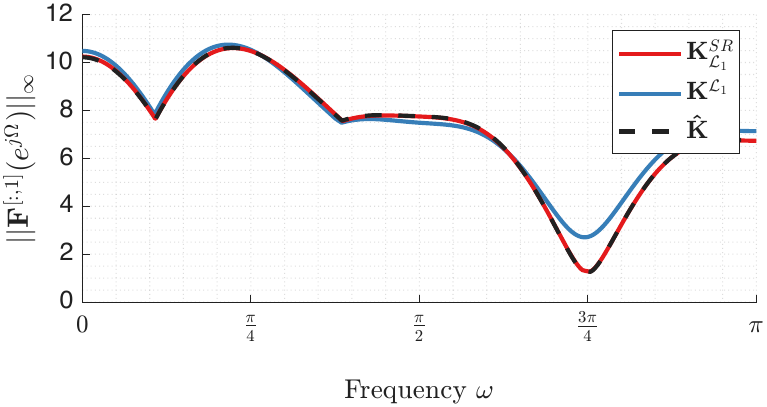}
    \caption{Plot of the $\infty$-norm of $\tf{F}_{\ell}^{[:,1]}(e^{j \omega})$ as a function of frequency.}
    \label{fig:both_plots_with_10_masses}
\end{figure}

\preprintswitch{}{
\begin{figure}[t]
    \centering
    \includegraphics[width=0.35\textwidth]{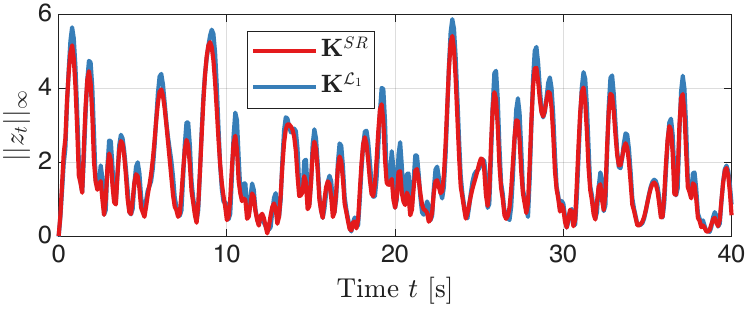}
    \caption{Plot of \(\norm{z_t}_\infty\) as a function of time for both controllers under the realization of a disturbance with $w^{[1]}_t = \sum_{k=0}^{100}\cos\left(\frac{1}{100}\left(\frac{3\pi}{4}(100-k) + \pi k \right) t + \phi_k\right)$, where $\phi_k$ is drawn uniformly from $[0, 2\pi)$.}
    \label{fig:response_10_masses_sinusoidal_disturbance}
\end{figure}
}

\section{Conclusion}\label{sec:conclusion}

This paper introduced spatial regret, a graph-informed metric for quantifying the performance loss caused by distributed information constraints relative to an oracle with enhanced sensing. Unlike classical $\mathcal{H}_2$ and $\mathcal{H}_\infty$ norms, spatial regret focuses on disturbances for which additional sensor information leads to the greatest improvement, such as localized disturbances affecting individual nodes.

We established well-posedness conditions to ensure the oracle provides a meaningful benchmark and developed two synthesis approaches: an exact SDP formulation for $q=2$ and a scalable distributed algorithm for $q=\infty$ tailored to large-scale systems. Numerical experiments on power grid models demonstrated 15-27\% performance improvements over classical controllers for localized multi-frequency disturbances, validating the effectiveness of our framework.

These results open several directions for future research. One avenue is the development of more efficient algorithms for synthesizing spatial regret-optimal controllers, avoiding SDPs for $q=2$ and improving upper bounds for $q=\infty$. Further, exploring the joint exploitation of non-causality and sparsity in oracle selection may yield additional benefits. Incorporating model uncertainty through robust spatial regret formulations or designing controllers based on data-driven frequency responses are also promising directions.

Overall, spatial regret represents a concrete step toward integrating network structure and information limitations into performance metrics, enabling targeted and effective synthesis of distributed controllers for complex networked systems.

\bibliography{bibliography}

\appendix

\subsection{Distributed control design for unstable plants}\label{app:review_when_plant_is_unstable}

This section reviews distributed control design for unstable plants. Once Assumption~\ref{ass:networked_system_stable} is removed, the existence of a controller $\tf{K} \in \mathcal{C}_{\text{stab}} \cap \netstruc(\graph)$ is no longer guaranteed. We therefore seek conditions ensuring the following assumption. 
\begin{assumption} \label{ass:NS(G)_non_empty}
    The set $\mathcal{C}_{\text{stab}} \cap \netstruc(\graph)$ is non-empty.
\end{assumption}

While necessary and sufficient conditions for the stabilizability of network-structured plants using network-structured controllers remains an open problem for future research, we provide sufficient conditions guaranteeing Assumption~\ref{ass:NS(G)_non_empty}. 
\begin{corollary}[\cite{lunze1992feedback}, Th. 4.8] \label{cor:stabilizability_conditions_via_dec_fix_modes}
Let $\mathcal{K}^{D}$ denote the set of static decentralized output controllers:
{\setlength{\abovedisplayskip}{4pt}%
    \setlength{\belowdisplayskip}{4pt}%
    \begin{equation*}
        \mathcal{K}^{D} \coloneq 
        \{
        K \in \real{\blue{m \times p}} | {K}^{[i,j]} = 0, \:\: \forall i, j = 1, \dots ,N, \:\: i \neq j 
        \}\,,
    \end{equation*}}
    and define the set of \emph{decentralized fixed modes} as \preprintswitch{
        $
        \Lambda_{df} \coloneq \bigcap\limits_{ {K} \in \mathcal{K}^{D}} \sigma ( A + B_{2} {K} C_2 )
    $.
    }{:
    \begin{equation*}
        \Lambda_{df} \coloneq \bigcap\limits_{ {K} \in \mathcal{K}^{D}} \sigma ( A + B_{2} {K} C_2 ) \,.
    \end{equation*}
    }
    If $\abs{\lambda_i} < 1$ for all $\lambda_i \in \Lambda_{df}$, then \blue{$\mathcal{C}_{\textnormal{stab}} \cap \netstruc(\graph)$ is non-empty, i.e., Assumption~\ref{ass:NS(G)_non_empty} holds}.
\end{corollary}

The proof of~\Cref{cor:stabilizability_conditions_via_dec_fix_modes} follows from~\cite[Th. 4.8]{lunze1992feedback}: when all decentralized fixed modes are stable, there exists a static decentralized controller $K \in \mathcal{K}^{D}$ that stabilizes the plant. \blue{Since $\mathcal{K}^{D} \subseteq \netstruc(\graph)$, this controller belongs to $\mathcal{C}_{\textnormal{stab}} \cap \netstruc(\graph)$, which is therefore non-empty}.
Decentralized fixed modes are eigenvalues that cannot be shifted by any decentralized controller. Standard methods exist to efficiently compute the set $\Lambda_{df}$ by studying the topology of the graph (e.g., \cite{pichai1984graph}). When the fixed modes are unstable (i.e., $\exists\abs{\lambda_i} > 1$), the information architecture should be enriched before meaningful performance optimization becomes possible.
Having established stabilizability, we develop a complete characterization of all stabilizing distributed controllers using the Youla parameterization for unstable plants. 
\begin{definition}%
A collection of stable transfer matrices $\tf{U}_l, \tf{V}_l, \tf{N}_l, \tf{M}_l, \tf{U}_r, \tf{V}_r, \tf{N}_r, \tf{M}_r \in \mathcal{RH}_\infty$ constitutes a \emph{doubly coprime factorization} of $\tf{P}_{22}$ if $ \tf{P}_{22} = \tf{N}_r\tf{M}_r^{-1} = \tf{M}_l^{-1}\tf{N}_l $ and:
{\setlength{\abovedisplayskip}{0pt}%
    \setlength{\belowdisplayskip}{2pt}%
\begin{equation*}
\begin{bmatrix} \tf{U}_l & -\tf{V}_l \\ -\tf{N}_l & \tf{M}_l\end{bmatrix}
\begin{bmatrix} \tf{M}_r & \tf{V}_r \\ \tf{N}_r & \tf{U}_r\end{bmatrix} = I.
\end{equation*}}
\end{definition}

The existence of such factorizations is guaranteed if $\tf{P}_{22}$ is stabilizable and detectable~\cite{zhou1998essentials}. This factorization enables the complete Youla parameterization~\cite{zheng2020equivalence}:{\setlength{\abovedisplayskip}{1pt}%
\begin{equation}\label{eq:youla_unstable_case}
\mathcal{C}_{\text{stab}} = \{\tf{K} = (\tf{V}_r - \tf{M}_r\tf{Q})(\tf{U}_r - \tf{N}_r\tf{Q})^{-1} \mid \tf{Q} \in \mathcal{RH}_\infty\}.
\end{equation}}
Using the controller parameterization~\eqref{eq:youla_unstable_case}, the LFT in~\eqref{eq:definition_LFT} can be rewritten in terms of $\tf{Q}$ as: 
\(
 \tf{F}_\ell(\tf{P}, \tf{Q}) = \tf{\tilde{P}}_{11} + \tf{\tilde{P}}_{12} \tf{Q} \tf{\tilde{P}}_{21},
 \)
where the transformed plant blocks are $\tf{\tilde{P}}_{11} \coloneq \tf{P}_{11} + \tf{P}_{12} \tf{V}_r \tf{M}_l \tf{P}_{21}$, $\tf{\tilde{P}}_{12} \coloneq -\tf{P}_{12} \tf{M}_r$, and $\tf{\tilde{P}}_{21} \coloneq \tf{M}_l \tf{P}_{21}$.
\begin{remark} \label{remark:pre_stabilization_with_F_and_L}
    Assumption~\ref{ass:NS(G)_non_empty} can also be satisfied without using~\Cref{cor:stabilizability_conditions_via_dec_fix_modes}, as follows. Given a matrix $F$ with block partition structure $\mathcal{A}(\graph)$ and a block-diagonal matrix $L$, such that $(A + B_2 F)$ and $(A + L C_2)$ are stable, define $\tf{K}_{0} \coloneq C_{K_{0}} (\textnormal{z}I - A_{K_{0}})^{-1} B_{K_{0}}$ with state-space matrices $A_{K_{0}} \coloneq (A + B_2 F + L C_2 + L D_{22} F)$, $B_{K_{0}} \coloneq -L$, and $C_{K_{0}} \coloneq F$. Observe that $\tf{K}_{0} \in \mathcal{C}_{\text{stab}} \cap \netstruc(\graph)$, thus Assumption~\ref{ass:NS(G)_non_empty} is satisfied. 
    Matrices $F$ and $L$ can be obtained by brute force or via SDP-based methods (e.g.,~\cite[Sec. VI.A]{vamsi2015optimal}).
    A doubly coprime factorization using $F$ and $L$ is reported in~\cite[Th. 11.8]{zhou1998essentials}. 
    \looseness=-1
\end{remark}

Once a doubly coprime factorization is obtained, all theoretical results in Sections~\ref{sec:problem_formulation}--\ref{sec:main_results} apply unchanged by considering the transformed plant $\tf{\tilde{P}}$. 
A well-known result in classical distributed control theory states that, to solve problem~\eqref{eq:classical_optimal_distributed_control_problem} in a convex manner, a necessary and sufficient condition is that the set $\netstruc(\graph)$ satisfies the QI condition~\cite{rotkowitz2005characterization}.
\begin{lemma}
    Given a plant $\tf{P} \in \netstruc(\graph)$, the set $\netstruc(\graph)$ is QI under $\graph$, i.e., $\tf{K} \tf{P}_{22} \tf{K} \in \netstruc(\graph)$ for all $\tf{K} \in \netstruc(\graph)$. 
\end{lemma}
\begin{proof}
 The set $\netstruc(\graph)$ is closed under products: for any $\tf{P}_{22},\tf{K} \in \netstruc(\graph)$, we have $\tf{P}_{22} \tf{K} \in \netstruc(\graph)$~\cite{naghnaeian2024youla}. Applying this property twice yields $ \tf{P}_{22} \tf{K} \in \netstruc(\graph)$ and, subsequently, $\tf{K}( \tf{P}_{22} \tf{K}) \in \netstruc(\graph)$.
\end{proof}
This result ensures that the set of stabilizing and network-structured controllers $\mathcal{C}_{\textnormal{stab}} \cap \netstruc(\graph)$ admits a convex reformulation in the Youla parameter~\cite[Sec. IV.C]{zheng2020equivalence}:
{\setlength{\abovedisplayskip}{4pt}%
    \setlength{\belowdisplayskip}{4pt}%
\begin{multline} \label{eq:sparsity_set_unstable_plant_youla}
    \mathcal{C}_{\textnormal{stab}} \cap \netstruc(\graph) = \{ \tf{K} = (\tf{V}_r - \tf{M}_r \tf{Q}) (\tf{U}_r - \tf{N}_r \tf{Q})^{-1}|
    \\
     (\tf{V}_r - \tf{M}_r \tf{Q}) \tf{M}_l \in \netstruc(\graph) , \tf{Q}\in \mathcal{RH}_{\infty}^{m \times p } \}.
\end{multline}}
In particular, when Assumption~\ref{ass:networked_system_stable} holds (i.e., $\tf{P}_{22}$ is stable), a doubly coprime factorization of $\tf{P}_{22}$ can be trivially chosen as:
$\tf{U}_{l} = - I$, $\tf{V}_{l} = 0$, $\tf{N}_{l} = \tf{P}_{22}$, $\tf{M}_{l} = I$, $\tf{U}_{r} = I$, $\tf{V}_{r} = 0$, $\tf{N}_{r} = -\tf{P}_{22}$, $\tf{M}_{r} = -I$,
which reduces~\eqref{eq:sparsity_set_unstable_plant_youla} to~\eqref{eq:definition_YK_parameterization}.

\subsection{Proof of Theorem~\ref{th:well_posedness}}
\label{app:proof_th_well_posedness_with_K}
Consider the metric $J_q({w},\tf{Q})$ defined in~\eqref{eq:definition_cost_J} for the closed-loop system affected by disturbance ${w}$ and controlled by Youla parameter $\tf{Q}$. Given the oracle $\tf{\hat{Q}}$ from~\eqref{eq:def_oracle_problem_in_Q}, we prove that:
{\setlength{\abovedisplayskip}{6pt}%
    \setlength{\belowdisplayskip}{6pt}%
    \begin{equation}\label{eq:th_1_step_1}
        \forall \tf{\bar{Q}} \in \netstruc(\hat{\graph})_{\textnormal{stab}}, \:\:\exists\norm{{w}}_q \blue{=}1  \: : \:J_q( {w} ,\tf{\bar{Q}}) \geq 
        J_q(\tf{w},\tf{\hat{Q}})\,.
    \end{equation}}
We proceed by contradiction. Suppose~\eqref{eq:th_1_step_1} is false, then:
{\setlength{\abovedisplayskip}{6pt}%
    \setlength{\belowdisplayskip}{6pt}%
    \begin{equation}\label{eq:contradicted_sentence_well_posed}
        \exists \tf{\bar{Q}} \in \netstruc(\hat{\graph})_{\textnormal{stab}}, \:\: \forall \norm{{w}}_q\blue{=} 1 \: : \: J_q({w},\tf{\bar{Q}}) < 
        J_q({w},\tf{\hat{Q}})\,.
    \end{equation}}
We analyze all three possible cases based on the oracle's design criterion.

\textbf{Case 1: $\H2$-optimal oracle ($q=2$).} By hypothesis, $\tf{\hat{Q}} \in \netstruc(\hat{\graph})_{\text{stab}}$ minimizes $\|\tf{F}_{\ell}(\tf{P},\tf{\hat{Q}})\|_{\mathcal{H}_2}^2$. From the theory of $\mathcal{H}_2$ optimal control~\cite[Ch. 4.3]{zhou1998essentials}, the oracle achieves optimality against the orthonormal set of impulse disturbances. Specifically:
        \(
       \norm{\tf{F}_{\ell}(\tf{P},\tf{\hat{Q}})}_{\H2}^2 \!\!\!\! = \!\! \sum_{i=1}^{m}\norm{\hat{z}_i}_2^2,
       \)
    where $\hat{z}_i$ is the performance output when an impulse is applied as a disturbance to the $i^{th}$ channel of $\tf{F}_{\ell}(\tf{P},\tf{\hat{Q}})$.
   Since \( \tf{\hat{Q}} \) is \( \H2 \)-optimal, it follows that:
   {\setlength{\abovedisplayskip}{5pt}%
    \setlength{\belowdisplayskip}{5pt}%
    \begin{equation}\label{eq:step_proof_h2_oracle_is_well_posed}
        \sum_{i=1}^{m} \norm{\hat{z}_i}_2^2 \leq \sum_{i=1}^{m} \norm{\bar{z}_i}_2^2 = \norm{\tf{F}_{\ell}(\tf{P}, \tf{\bar{Q}})}_{\H2}^2,
    \end{equation}}
     where $\bar{z}_i$ is the performance output when the same impulsive disturbances are applied to $\tf{F}_{\ell}(\tf{P},\tf{\bar{Q}})$.
    From~\eqref{eq:step_proof_h2_oracle_is_well_posed}, there must exist some index $i^* \in \{1, \ldots, m\}$ such that $\|{\hat{z}}_{i^*}\|_2^2 \leq \|{\bar{z}}_{i^*}\|_2^2$. Therefore, for the impulse disturbance applied to channel $i^*$, we have $J_2({w}_{i^*}, \tf{\hat{Q}}) \leq J_2({w}_{i^*}, \tf{\bar{Q}})$, which contradicts~\eqref{eq:contradicted_sentence_well_posed}.

\textbf{Case 2: $\Hinfty$-optimal oracle ($q=2$).}
    If $\tf{\hat{Q}}$ minimizes the $\Hinfty$ norm, i.e., $\max_{\lVert {w} \rVert_{2} \leq 1}[J_2({w}, \tf{\hat{Q}})]$, then for any $\norm{w}_2 \blue{=} 1$:
    {\setlength{\abovedisplayskip}{6pt}%
    \setlength{\belowdisplayskip}{6pt}%
\begin{equation}\label{eq:intermediate_step_Hinf_proof_well_posed}
        \max_{\lVert {w}_1 \rVert_{2} \blue{=} 1}[J_2({w}_1, \tf{\bar{Q}})]
        \geq 
        \max_{\lVert {w}_2 \rVert_{2} \blue{=} 1}[J_2({w}_2, \tf{\hat{Q}})]
        \geq
        J_2({w}, \tf{\hat{Q}}).
    \end{equation} }
    Let $w^* \in \arg\max_{\lVert {w}_1 \rVert_{2} \blue{=} 1}[J_2({w}_1, \blue{\tf{\bar{Q}}})]$. Then~\eqref{eq:intermediate_step_Hinf_proof_well_posed} implies that 
        \(
        J_2(w^*, \tf{\bar{Q}})
        \geq
        J_2(w^*, \tf{\hat{Q}}),
        \)
    which contradicts~\eqref{eq:contradicted_sentence_well_posed}.

\textbf{Case 3: $\mathcal{L}_1$-optimal oracle ($q=\infty$).} The proof proceeds analogously to that of Case 2, considering the supremum norm over $\ell_\infty$ disturbances, instead of $\ell_2$ ones.

Therefore,~\eqref{eq:th_1_step_1} must hold for any oracle obtained as in~\eqref{eq:def_oracle_problem_in_Q}. Given \eqref{eq:th_1_step_1}, it immediately follows that for $q\in \{2, \infty\}$:
{\setlength{\abovedisplayskip}{1pt}%
    \setlength{\belowdisplayskip}{3pt}%
    \begin{multline}\label{eq:last_step_firts_preposition}
        \forall \tf{Q} \in \netstruc(\graph)_{\textnormal{stab}} \subseteq \netstruc(\hat{\graph})_{\textnormal{stab}} \rightarrow 
        \\
        \exists {w}, \:\norm{{w}}_q\blue{=}1 \: : \:J_q({w},\tf{Q}) -
        J_q({w},\tf{\hat{Q}})\geq0\,.
     \end{multline}}
Using the definition~\eqref{eq:def_spatial_regret_cost} of spatial regret,~\eqref{eq:last_step_firts_preposition} establishes that $\text{SpReg}_q(\tf{Q}, \tf{\hat{Q}}) \geq 0$ for any $\tf{Q} \in \netstruc(\graph)_{\text{stab}}$ and $q\in \{2, \infty\}$, completing the proof.

\subsection{Proof of Lemma~\ref{lemma:Hinf_with_l2_signals}}\label{app:proof_lemma_Hinf_with_l2_signals}

Let $z$ and $\hat{z}$ denote the outputs of the two closed-loop systems $\tf{F}_{\ell}(\tf{\tilde{P}}, \tf{Q})$ and $\tf{F}_{\ell}(\tf{\tilde{P}}, \tf{\hat{Q}})$, respectively, when affected by the same disturbance $w$ with $\norm{w}_2 \blue{= 1}$. Using the definition~\eqref{eq:def_spatial_regret_cost} and Parseval's theorem, we derive:
\allowdisplaybreaks
\begin{align*}
&\spreg_2(\tf{Q}, \tf{\hat{Q}}) %
=
\blue{\sup}_{\norm{w}_2 \blue{= 1}}
\frac{1}{2 \pi} \int_{- \pi}^{+ \pi} \!\!\!\tf{w}^*(e^{j \omega})
\tf{\Psi}(e^{j \omega})
 \tf{w}(e^{j \omega}) d\omega
\\
&\leq
\blue{\sup}_{\norm{w}_2 \blue{= 1}}
\frac{1}{2 \pi} \int_{- \pi}^{+ \pi}\lambda_{\textnormal{max}}\big( \tf{\Psi}(e^{j \omega})
\big) 
\tf{w}^*(e^{j \omega})
 \tf{w}(e^{j \omega}) d\omega 
\\
& \leq 
\sup_{\omega \in \Omega }\lambda_{\textnormal{max}}\big( \tf{\Psi}(e^{j \omega}) \big)
\blue{\sup}_{\norm{w}_2 \blue{= 1}}
\frac{1}{2 \pi} \int_{- \pi}^{+ \pi} \tf{w}^*(e^{j \omega})\tf{w}(e^{j \omega}) d\omega
\\
&= 
\sup_{\omega \in \Omega }\lambda_{\textnormal{max}}\big( \tf{\Psi}(e^{j \omega}) \big) \blue{\sup}_{\norm{w}_2 \blue{= 1}} \norm{w}_2^2
= 
\sup_{\omega \in \Omega }\lambda_{\textnormal{max}}\big( \tf{\Psi}(e^{j \omega}) \big) \,.
\end{align*}
To complete the proof, we construct \blue{a family of} disturbance\blue{s} $\bar{w}$ with $\norm{\bar{w}}_2 \blue{= 1}$ \blue{whose performance gap approaches the supremum as $\varepsilon \to 0$}.
Let $\omega_0 \in \arg \sup_{\omega \in \Omega} \lambda_{\textnormal{max}} \big( \tf{\Psi}(e^{j\omega}) \big)$, assuming $\omega_0 \neq \pm \pi$. Since $\tf{\Psi}(e^{j\omega_0})$ is Hermitian, it admits the decomposition:
\(
\tf{\Psi}(e^{j \omega_0}) = 
\lambda_{\textnormal{max}} v_1(e^{j \omega_0})v_1^*(e^{j \omega_0}) + \sum_{i=2}^{\psi} \lambda_{i}v_i(e^{j \omega_0})v_i^*(e^{j \omega_0})\,,
\)
where $\psi$ is the rank of $\tf{\Psi}(e^{j \omega_0})$ and $v_i\in \mathbb{C}^{n_w}$ have unit length.
Ideally, using $v_1(e^{j\omega_0})$ as the disturbance $\bar{w}$ would suffice, but a real-valued signal is required. To address this, write the definition of %
\preprintswitch{
    \(
    v_1(e^{j\omega_0}) = 
    \begin{bmatrix}
        \alpha_1 e^{j\theta_1}
        &
        \cdots
        &
        \alpha_{n_w} e^{j\theta_{n_w}}
    \end{bmatrix}^\top,
\)
}{
\begin{equation*}
    v_1(e^{j\omega_0}) = 
    \begin{bmatrix}
        \alpha_1 e^{j\theta_1}
        &
        \cdots
        &
        \alpha_{n_w} e^{j\theta_{n_w}}
    \end{bmatrix}^\top \,,
\end{equation*}
}
where for consistency $\alpha_i \in \real{}$ is such that $\theta_i \in (-\pi, 0 ]$ for any $i = 1 \dots n_w$.
Denote $\gamma \coloneq \frac{\omega_0}{\tan{(\nicefrac{\omega_0}{2}})}$ and $\beta_i \geq 0 $ such that:
\preprintswitch{
    \(
    \theta_i = \angle{
    \left(
    \frac{(\beta_i - \gamma)e^{j\omega_0} + (\beta_i + \gamma)}{(\beta_i + \gamma)e^{j\omega_0} + (\beta_i - \gamma)}
    \right)
    },
\)
}{
\begin{equation*}
    \theta_i = \angle{
    \left(
    \frac{(\beta_i - \gamma)e^{j\omega_0} + (\beta_i + \gamma)}{(\beta_i + \gamma)e^{j\omega_0} + (\beta_i - \gamma)}
    \right)
    }\,,
\end{equation*}
}
($\beta_i \to \infty$ if $\theta_i = 0$). 
Construct a scalar transfer function $\tf{f}_{f}(\textnormal{z})$ which (approximately) filters all the frequencies except $\omega_0$:
\begin{equation*}
    |\tf{f}_{f}(e^{j\omega}))| = 
    \left\{
\begin{aligned}
    \sqrt{\frac{\pi}{2\varepsilon}} \quad &\textnormal{if } |\omega - \omega_0| < \varepsilon \textnormal{ or } |\omega + \omega_0| < \varepsilon\,,
    \\
    0 \quad &\textnormal{otherwise.}
\end{aligned}
    \right.
\end{equation*}
where $\varepsilon >0$ is small.
Then, pick the disturbance $\bar{w}$ as:
   {\setlength{\belowdisplayskip}{1pt}%
\begin{equation}\label{eq:weird_disturbance_worst_case}
    \tf{\bar{w}}(\textnormal{z}) \coloneq
    \underbrace{
    \begin{bmatrix}
        \alpha_1 \frac{(\beta_1 - \gamma)\textnormal{z} + (\beta_1 + \gamma)}{(\beta_1 + \gamma)\textnormal{z} + (\beta_1 - \gamma)}
        \\
        \vdots
        \\
        \alpha_{n_w} \frac{(\beta_{n_w} - \gamma)\textnormal{z} + (\beta_{n_w} + \gamma)}{(\beta_{n_w} + \gamma)\textnormal{z} + (\beta_{n_w} - \gamma)}
    \end{bmatrix}
    }_{\coloneq \tf{f}_{v_{1}}( \textnormal{z} ) }
    \tf{f}_{f}(\textnormal{z})\,.
\end{equation}}
The function $\tf{f}_{v_{1}}$ in~\eqref{eq:weird_disturbance_worst_case} is chosen ad hoc such that $\tf{f}_{v_{1}}(e^{j\omega_0}) = v_1(e^{j\omega_0})$.
Substituting $\tf{\bar{w}}(z)$ as in~\eqref{eq:weird_disturbance_worst_case} into the definition of $\spreg_2$, it follows that:
\begin{align*}
&\norm{z}_2^2 - \norm{\hat{z}}_2^2 
= 
\frac{1}{2 \pi} \int_{- \pi}^{+ \pi} \tf{\bar{w}}^*(e^{j \omega})
\tf{\Psi}(e^{j \omega})
 \tf{\bar{w}}(e^{j \omega}) 
 d\omega 
 \\
& \blue{\xrightarrow{\varepsilon \to 0}}
\lambda_{\textnormal{max}}\big( \tf{\Psi}(e^{j \omega_0})) =\sup_{\omega \in \Omega }\lambda_{\textnormal{max}}\big( \tf{\Psi}(e^{j \omega}) \big) ,
\end{align*}
which proves $\spreg_2(\tf{Q}, \tf{\hat{Q}}) = \sup_{\omega \in \Omega }\lambda_{\textnormal{max}}\big( \tf{\Psi}(e^{j \omega}) \big) $. For the case $\omega_0 = \pm \pi$, the intervals in $\tf{f}_f(z)$ can be easily adjusted to ensure compatibility with $\omega \in \Omega$, completing the proof.

\preprintswitch{}
{
\subsection{Proof of Theorem~\ref{thm:distributed_solution_ADMM}}\label{app:proof_theorem_distributed_solution_ADMM}

Problem~\eqref{eq:SLP_problem} can be rearranged as follows:
\begin{equation}\label{eq:ADMM_problem}
\min_{\tf{\Phi}, \tf{\Psi}}   \norm{\tf{\Delta}( \tf{\Phi} , \tf{\hat{\Phi}} )}_{\mathcal{L}_1} 
\text{s.t.} \:\:  \tf{\Phi} \in \mathcal{S}^{(r)}, \:\: \tf{\Psi} \in \mathcal{S}^{(c)}, \:\: \tf{\Phi} = \tf{\Psi}.
\end{equation}
Applying ADMM~\cite{anderson2019system} to~\eqref{eq:ADMM_problem} yields the iterations:
\begin{subequations}\label{eq:ADMM_iterative_steps_proof}
    \begin{align}
        \tf{\Phi}_{k+1} &\gets \! \argmin_{\tf{\Phi} \in \mathcal{S}^{(r)}} \! \norm{\tf{\Delta}( \tf{\Phi} , \tf{\hat{\Phi}} )}_{\mathcal{L}_1} \!\!\!\!+ \frac{\rho}{2} \!\norm{\tf{\Phi} \! - \! \tf{\Psi}_k \! + \! \tf{\Lambda}_k}_{\H2}^2 \! ,   \label{eq:ADMM_iterative_steps_row_update}
        \\
        \tf{\Psi}_{k+1} &\gets \argmin_{\tf{\Psi} \in \mathcal{S}^{(c)}} \frac{\rho}{2} \norm{\tf{\Psi} - \tf{\Phi}_{k+1} - \tf{\Lambda}_k}_{\H2}^2,   \label{eq:ADMM_iterative_steps_column_update}
        \\
        \tf{\Lambda}_{k+1} &\gets \tf{\Lambda}_k + \tf{\Phi}_{k+1} - \tf{\Psi}_{k+1}. \label{eq:ADMM_iterative_steps_dual_update}
    \end{align}
\end{subequations}
The convergence of ADMM follows from problem convexity and feasible set non-emptiness~\cite{boyd2011distributed,anderson2019system}. 
Each step in~\eqref{eq:ADMM_iterative_steps_proof} can be solved in a distributed manner, as we now explain.
Problem~\eqref{eq:ADMM_iterative_steps_column_update} exhibits column separability: both the constraint set $\mathcal{S}^{(c)}$ and the $\H2$ norm are column-separable. This allows decomposition into $M$ independent subproblems, yielding Steps~\ref{alg_step:for_loop_col}--\ref{alg_step:column_wise_auxiliary_update} via the definition of $\mathcal{F}^{(c)}_j$ in~\eqref{eq:F_c_functional_definition}.
The dual update~\eqref{eq:ADMM_iterative_steps_dual_update} is a sum that requires no distributed computation.
The main challenge lies in problem~\eqref{eq:ADMM_iterative_steps_row_update}, which contains the non-separable cost $\norm{\tf{\Delta}( \tf{\Phi} , \tf{\hat{\Phi}} )}_{\mathcal{L}_1}\!\!\!.$
To enable a distributed solution, we introduce an epigraph variable $\gamma \geq 0$ and reformulate~\eqref{eq:ADMM_iterative_steps_row_update} as:
\begin{equation}\label{eq:ADMM_iterative_steps_row_update_reformulated_in_gamma}
    \min_{\gamma}
    \gamma +\frac{\rho}{2}
    \Bigg(
    \underbrace{
    \begin{aligned}
        \min_{\tf{\Phi} \in \mathcal{S}^{(r)}}  \norm{\tf{\Phi} - \tf{\Psi}_k + \tf{\Lambda}_k}_{\H2}^2
        \\
        \text{s.t.} \quad \norm{\tf{\Delta}( \tf{\Phi} , \tf{\hat{\Phi}} )}_{\mathcal{L}_1} \leq \gamma
    \end{aligned}}_{=\min_{\tf{\Phi}}\mathcal{F}^{(r)}(\gamma , \tf{\Phi}, \tf{\Psi}_k, \tf{\Lambda}_k, \tf{\hat{\Phi}})}
    \Bigg).
\end{equation}
Under Assumption~\ref{ass:C1_D12_block_diagonal_permutation}, the constraint \(\|\tf{\Delta}(\tf{\Phi},\tf{\hat{\Phi}})\|_{\mathcal{L}_1}\le\gamma\) is row-separable (for more details, see~\cite[Ex. 14]{anderson2019system}).
This enables parallel minimization of $\mathcal{F}^{(r)}$ over $M$ independent subproblems corresponding to the row blocks $r_i$ of $\tf{\Phi}$, as shown in Algorithm~\ref{alg:parallel_min_F(r)}.
The outer optimization~\eqref{eq:ADMM_iterative_steps_row_update_reformulated_in_gamma} over the scalar $\gamma$ (Step~\ref{alg_step:gamma_minimization_golden_search} in Algorithm~\ref{alg:ADMM_spreg}) can be efficiently solved using golden-section search or simply the bisection method.
}

\preprintswitch{}{
\subsection[Computation of a lower-bound to SpReg-inf]{Computation of a lower bound to $\spreg_\infty$}\label{app:lower_bound_spreg_infty}

We complement~\Cref{prop:LP_reformulation_finite_dim} with a computable lower bound on $\spreg_\infty(\tf{Q},\tf{\hat{Q}})$, so that the conservatism of the $\mathcal{L}_1$ upper bound can be assessed numerically for any given pair $(\tf{Q},\tf{\hat{Q}})$. The bound is obtained by restricting the disturbances to a finite support, in which case the supremum in~\eqref{eq:def_spatial_regret_cost} can be evaluated exactly by solving a finite number of LPs.

As in the proof of~\Cref{prop:LP_reformulation_finite_dim}, let $z$ and $\hat{z}$ denote the outputs of $\tf{F}_\ell(\tf{P},\tf{Q})$ and $\tf{F}_\ell(\tf{P},\tf{\hat{Q}})$ under a common disturbance $w$, that is, for $i \in \{1,\dots,n_z\}$ and $t \geq 0$,
\begin{equation}\label{eq:lb_convolutions}
  z_t(i) = \sum_{j=1}^{n_w}\sum_{\tau=0}^{t} F_\ell^{[i,j]}(t-\tau)\, w_\tau(j)\,,
\end{equation}
and analogously for $\hat{z}_t(i)$ with $\hat{F}_\ell^{[i,j]}(t)$ in place of $F_\ell^{[i,j]}(t)$. We assume that both closed-loop maps are FIR of order $N_{\textnormal{FIR}}$, i.e.,
\begin{equation}\label{eq:lb_fir_assumption}
  F_\ell^{[i,j]}(t) = \hat{F}_\ell^{[i,j]}(t) = 0\,, \quad \forall t > N_{\textnormal{FIR}}\,, \:\: \forall i,j\,,
\end{equation}
which is true when the closed-loop responses are parameterized through the SLP~\eqref{eq:SLP_parametrization_closed_loop_map} with FIR $\tf{\Phi}$ (see~\Cref{subsec:distributed_computation_q_infty}). Moreover, since $J_\infty(w,\tf{Q})$ is positively homogeneous in $w$ and $\spreg_\infty(\tf{Q},\tf{\hat{Q}}) \geq 0$ by~\Cref{th:well_posedness}, the supremum in~\eqref{eq:def_spatial_regret_cost} is unchanged when the constraint $\norm{w}_\infty = 1$ is relaxed to $\norm{w}_\infty \leq 1$.
We adopt the latter convex formulation in what follows.
\looseness=-1

Restrict the disturbance support to $\{0,\dots,T\}$ for some $T \geq 0$. Reusing, with a slight abuse of notation, the symbol $w$ for the stacked vector $\mathrm{vec}(w_0,\dots,w_T)$, we define the set of admissible disturbances $\mathcal{W}_T \coloneq \{w \in \real{n_w(T+1)} : \norm{w}_\infty \leq 1\}$. Then, we can define a restricted version of the spatial regret over this set of disturbances:
\begin{equation}\label{eq:lb_truncated_regret}
  \spreg_\infty^{(T)}(\tf{Q},\tf{\hat{Q}})
  \coloneq \max_{w \in \mathcal{W}_T} \big[\norm{z}_\infty - \norm{\hat{z}}_\infty\big]\,.
\end{equation}
Since every $w \in \mathcal{W}_T$ is a feasible solution for~\eqref{eq:def_spatial_regret_cost}, it holds that $\spreg_\infty^{(T)}(\tf{Q},\tf{\hat{Q}}) \leq \spreg_\infty(\tf{Q},\tf{\hat{Q}})$ for every $T \geq 0$, and the bound becomes tighter as $T$ increases. Moreover, by~\eqref{eq:lb_convolutions} and~\eqref{eq:lb_fir_assumption}, the outputs satisfy $z_t(i) = \hat{z}_t(i) = 0$ for all $t > T_z \coloneq T + N_{\textnormal{FIR}}$, so the norms in~\eqref{eq:lb_truncated_regret} only involve the finite vectors $z = \mathrm{vec}(z_0,\dots,z_{T_z})$ and $\hat{z} = \mathrm{vec}(\hat{z}_0,\dots,\hat{z}_{T_z})$.

It remains to show that~\eqref{eq:lb_truncated_regret} can be computed exactly. 
In the following proposition, we demonstrate  that~\eqref{eq:lb_truncated_regret} coincides with the largest optimal value of a finite set of LPs.

\begin{proposition}\label{prop:lb_enumeration}
For each $(i,t) \in \mathcal{J} \coloneq \{1,\dots,n_z\} \times \{0,\dots,T_z\}$, let $a_{i,t}, b_{i,t} \in \real{n_w(T+1)}$ be the vectors satisfying:
\begin{equation}\label{eq:lb_linear_maps}
  z_t(i) = a_{i,t}^\top w\,, \qquad \hat{z}_t(i) = b_{i,t}^\top w\,,
\end{equation}
for every $w \in \mathcal{W}_T$, and for each $(i,t,\sigma) \in \mathcal{J} \times \{-1,+1\}$ consider the LP:
\begin{equation}\label{eq:lb_final_LP}
  \begin{aligned}
    v(i,t,\sigma) \coloneq
    \max_{w \in \mathcal{W}_T,\, \hat{\gamma}}
    \quad & \sigma a_{i,t}^\top w - \hat{\gamma} \\[2pt]
    \st \quad
    & \hat{\gamma} \geq \pm\, b_{j,s}^\top w\,, \quad \forall (j,s) \in \mathcal{J}\,.
  \end{aligned}
\end{equation}
Then, it holds that:
\begin{equation}\label{eq:lb_enumeration}
  \spreg_\infty^{(T)}(\tf{Q},\tf{\hat{Q}}) =
  \max_{\substack{(i,t) \in \mathcal{J} \\ \sigma \in \{-1,+1\}}} v(i,t,\sigma)\,.
\end{equation}
\end{proposition}

\begin{proof}
The vectors in~\eqref{eq:lb_linear_maps} are well defined because, by~\eqref{eq:lb_convolutions}, each sample $z_t(i)$ depends linearly on the stacked disturbance: the subvector of $a_{i,t}$ multiplying $w_\tau$ collects the coefficients $F_\ell^{[i,j]}(t-\tau)$ for $j = 1,\dots,n_w$ and is zero whenever $\tau > \min(t,T)$, and $b_{i,t}$ is built analogously from $\hat{F}_\ell^{[i,j]}$. Since the outputs vanish for $t > T_z$, every $w \in \mathcal{W}_T$ satisfies $\norm{z}_\infty = \max_{(i,t) \in \mathcal{J}} \abs{a_{i,t}^\top w}$ and $\norm{\hat{z}}_\infty = \max_{(j,s) \in \mathcal{J}} \abs{b_{j,s}^\top w}$. Substituting into~\eqref{eq:lb_truncated_regret} and writing $\abs{a_{i,t}^\top w} = \max_{\sigma \in \{-1,+1\}} \sigma a_{i,t}^\top w$ yields:
\begin{multline}\label{eq:lb_dc_form}
  \spreg_\infty^{(T)}(\tf{Q},\tf{\hat{Q}}) = \\
  \max_{w \in \mathcal{W}_T}\:
  \max_{\substack{(i,t) \in \mathcal{J} \\ \sigma \in \{-1,+1\}}}
  \Big[\sigma a_{i,t}^\top w - \max_{(j,s) \in \mathcal{J}} \abs{b_{j,s}^\top w}\Big]\,,
\end{multline}
where the term $(-\max_{(j,s)} \abs{b_{j,s}^\top w})$, which does not depend on $(i,t,\sigma)$, has been moved inside the maximum operator and where we exchanged the order of the two maximizations in~\eqref{eq:lb_dc_form}.
For fixed $(i,t,\sigma)$ the inner problem maximizes the concave function $\left[\sigma a_{i,t}^\top w - \max_{(j,s)} \abs{b_{j,s}^\top w}\right]$ over the polytope $\mathcal{W}_T$; introducing the epigraph variable $\hat{\gamma}$ together with the constraints $\hat{\gamma} \geq \pm b_{j,s}^\top w$, which enforce $\hat{\gamma} \geq \max_{(j,s)} \abs{b_{j,s}^\top w}$ with equality at the optimum since the objective is decreasing in $\hat{\gamma}$, casts it as the LP~\eqref{eq:lb_final_LP}. Taking the maximum over the finitely many triples establishes~\eqref{eq:lb_enumeration}.
\end{proof}

Each LP in~\eqref{eq:lb_final_LP} has $n_w(T+1) + 1$ decision variables and $2 n_w (T+1) + 2 n_z (T_z+1)$ inequality constraints, and is always feasible since $(w,\hat{\gamma}) = (0,0)$ satisfies all constraints. By~\Cref{prop:lb_enumeration}, the lower bound $\spreg_\infty^{(T)}(\tf{Q},\tf{\hat{Q}})$ is thus obtained by solving the $2 n_z (T_z+1)$ independent LPs in~\eqref{eq:lb_final_LP}, as summarized in~\Cref{alg:lower_bound_spreg}. The procedure also returns a worst-case disturbance: at the maximizing triple $(i^\star,t^\star,\sigma^\star)$, the LP optimizer $w_{\textnormal{opt}}$ satisfies $\sigma^\star a_{i^\star,t^\star}^\top w_{\textnormal{opt}} = \norm{z}_\infty$, so that $\norm{z}_\infty - \norm{\hat{z}}_\infty = \spreg_\infty^{(T)}(\tf{Q},\tf{\hat{Q}})$.

\begin{algorithm}[t] 
\caption{Computation of the lower bound $\spreg_\infty^{(T)}(\tf{Q},\tf{\hat{Q}})$.}
\label{alg:lower_bound_spreg}
\begin{algorithmic}[1]
\Require plant $\tf{P}$; controller $\tf{Q} \in \netstruc(\graph)_{\textnormal{stab}}$; oracle $\tf{\hat{Q}} \in \netstruc(\hat{\graph})_{\textnormal{stab}}$; support horizon $T \geq 0$
\State $T_z \gets T + N_{\textnormal{FIR}}$
\State Compute the impulse responses $F_\ell^{[i,j]}(t)$ and $\hat{F}_\ell^{[i,j]}(t)$ for $i \in \{1,\dots,n_z\}$, $j \in \{1,\dots,n_w\}$, $t \in \{0,\dots,T_z\}$
\State Assemble $a_{i,t},\, b_{i,t}$ for each $(i,t) \in \mathcal{J}$, as in~\eqref{eq:lb_linear_maps}
\State \textbf{for} each $(i,t,\sigma) \in \mathcal{J} \times \{-1,+1\}$ \textbf{in parallel do:}
    \State \hspace{0.5em} Solve the LP~\eqref{eq:lb_final_LP} to obtain $v(i,t,\sigma)$ and its optimizer $w(i,t,\sigma)$
\State $v_{\textnormal{opt}} \gets \max_{(i,t,\sigma)} v(i,t,\sigma)$, with corresponding $w_{\textnormal{opt}}$
\Ensure $\spreg_\infty^{(T)}(\tf{Q},\tf{\hat{Q}}) = v_{\textnormal{opt}}$ ;\: $w_{\textnormal{opt}}$\,.
\end{algorithmic}
\end{algorithm}
}

\preprintswitch{
\begin{IEEEbiography}[{\includegraphics[width=1in,height=1.25in,clip,keepaspectratio]{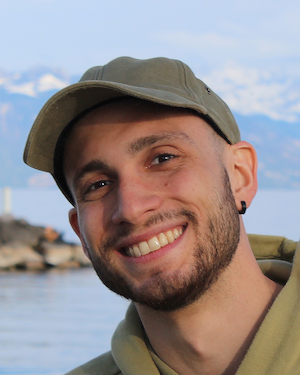}}]{Daniele Martinelli}
received the B.Sc. degree in automation and control engineering from the Università Federico II di Napoli, Naples, Italy, in 2020, and the M.Sc. degree in automation and control engineering from Politecnico di Milano, Milan, Italy, in 2022. He is currently pursuing a Ph.D. degree in robotics, control, and intelligent systems with the Automatic Control Laboratory at EPFL, Lausanne, Switzerland.
His research interests include networked control theory, optimization, and machine learning.
\end{IEEEbiography}
\begin{IEEEbiography}[{\includegraphics[width=1.1in,height=1.27in,clip]{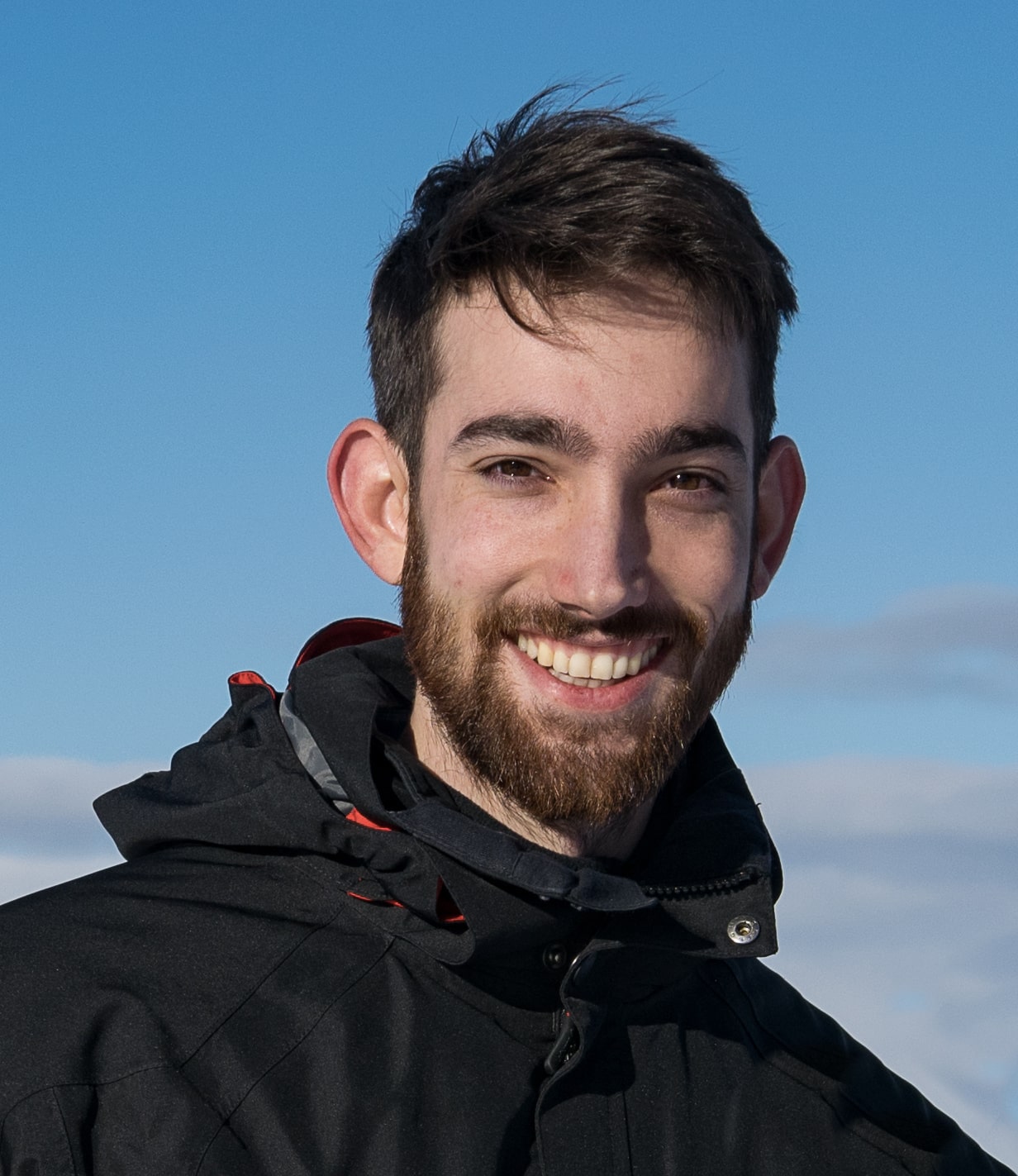}}]
{Andrea Martin}
 is a Digital Futures Postdoctoral Fellow at KTH Royal Institute of Technology. He received the Ph.D. degree in Robotics, Control, and Intelligent Systems from EPFL in 2025. During his doctoral studies, he was affiliated with the NCCR Automation. Previously, he obtained a B.Sc. in Information Engineering and two M.Sc. degrees in Automation Engineering and Automatic Control and Robotics from the University of Padova and the Polytechnic University of Catalonia through the TIME double degree program. His research interests include control theory, optimization, and machine learning.
\end{IEEEbiography}
\begin{IEEEbiography}[{\includegraphics[width=1.1in,height=1.27in,clip]{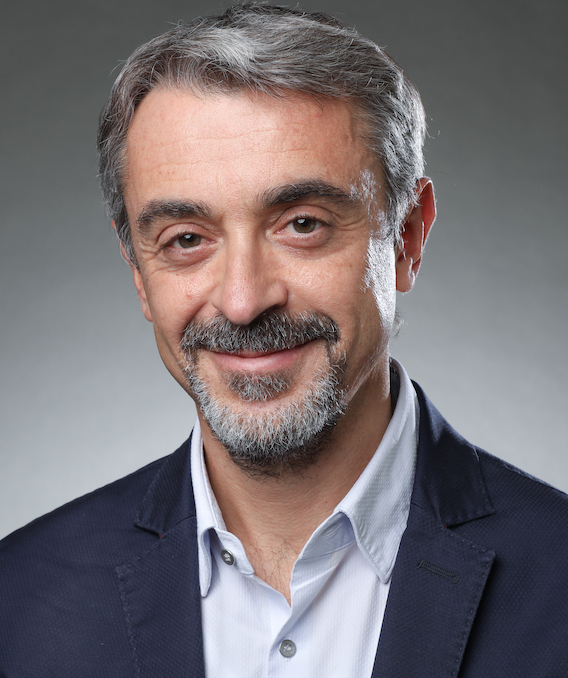}}]
{Giancarlo Ferrari-Trecate}
    (SM'12) received a Ph.D. in Electronic and Computer Engineering from the Università degli Studi di Pavia in 1999. Since September 2016, he has been a Professor at EPFL, Lausanne, Switzerland. In the spring of 1998, he was a Visiting Researcher at the Neural Computing Research Group, University of Birmingham, UK. In the fall of 1998, he joined the Automatic Control Laboratory, ETH, Zurich, Switzerland, as a Postdoctoral Fellow. He was appointed Oberassistent at ETH in 2000. In 2002, he joined INRIA, Rocquencourt, France, as a Research Fellow. From March to October 2005, he was a researcher at the Politecnico di Milano, Italy. From 2005 to August 2016, he was Associate Professor at the Dipartimento di Ingegneria Industriale e dell'Informazione of the Università degli Studi di Pavia.
    His research interests include scalable control, machine learning, microgrids, networked control, and hybrid systems.
    He and his team were winners or finalists for Best Student Paper awards at the European Control Conference (2023), the Conference on Decision and Control (2024), the Symposium of the European Association for Research in Transportation (2025), and the IEEE CSS Swiss Chapter Young Author Best Journal Paper Award (2024,2025).
    Giancarlo Ferrari Trecate is the founder and current chair of the Swiss chapter of the IEEE Control Systems Society. He is Senior Editor of the IEEE Transactions on Control Systems Technology and has served on the editorial boards of Automatica and Nonlinear Analysis: Hybrid Systems.
\end{IEEEbiography}
\begin{IEEEbiography}[{\includegraphics[width=1.1in,height=1.27in]{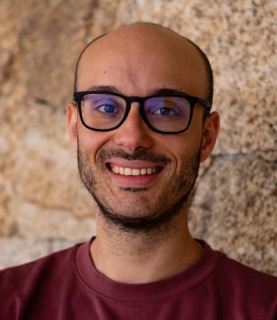}}]
{Luca Furieri} 
    (S'17, M'21) is an Associate Professor at the Department of Engineering Sciences of the University of Oxford. Previously, he served as a Principal Investigator at the Swiss Federal Institute of Technology in Lausanne (EPFL), supported by an Ambizione career grant from the Swiss National Science Foundation. He has received a Laurea degree from the University of Bologna in 2016 and a PhD in Electrical Engineering and Information Technology from ETH Zurich in 2020.  His papers have been awarded the IEEE Transactions on Control of Network Systems Best Paper Award in 2022, the European Control Conference Best Paper Award (finalist) in 2019, and the American Control Conference O. Hugo Schuck Best Paper Award in 2018. His group develops learning-based approaches that preserve stability, safety, and robustness in large-scale, interconnected systems such as energy and traffic networks.  
\end{IEEEbiography}
}{
}

\addtolength{\textheight}{-12cm}
\end{document}